\documentclass[journal]{IEEEtran}
\usepackage{amsmath}
\usepackage{amsfonts}
\usepackage{amssymb}
\usepackage{algorithm}
\usepackage{algorithmic}
\usepackage{graphicx}
\usepackage{verbatim}
\usepackage{amsthm}
\usepackage{xcolor}
\usepackage{multirow}
\usepackage{float}

\newtheorem{theorem}{Theorem}
\newtheorem{lemma}{Lemma}
\newtheorem{remark}{Remark}

\newtheorem{example}{Example}

\newcommand{\beq}{\begin{equation}}
\newcommand{\eeq}{\end{equation}}
\newcommand{\beqnn}{\begin{equation*}}
\newcommand{\eeqnn}{\end{equation*}}
\newcommand{\beqy}{\begin{eqnarray}}
\newcommand{\eeqy}{\end{eqnarray}}
\newcommand{\beqynn}{\begin{eqnarray*}}
\newcommand{\eeqynn}{\end{eqnarray*}}
\newcommand{\bit}{\begin{itemize}}
\newcommand{\eit}{\end{itemize}}
\newcommand{\ben}{\begin{enumerate}}
\newcommand{\een}{\end{enumerate}}
\newcommand{\bex}{\begin{example}}
\newcommand{\eex}{\end{example}}

\newcommand{\balg}[1]{\begin{algorithm} \caption{#1}}
\newcommand{\ealg}{\end{algorithm}}

\newcommand{\balgc}{\begin{algorithmic}[1]}
\newcommand{\ealgc}{\end{algorithmic}}

\newcommand{\bary}{\begin{array}}
\newcommand{\eary}{\end{array}}
\newcommand{\bmx}{\begin{bmatrix}}
\newcommand{\emx}{\end{bmatrix}}
\newcommand{\bsmx}{\left[\begin{smallmatrix}}
\newcommand{\esmx}{\end{smallmatrix}\right]}
\newcommand{\bmxc}[1]{\left[\begin{array}{@{}#1@{}}}
\newcommand{\emxc}{\end{array}\right]}
\newcommand{\bcn}{\begin{center}}
\newcommand{\ecn}{\end{center}}

\newcommand{\Rbb}{{\mathbb{R}}}
\newcommand{\Zbb}{{\mathbb{Z}}}

\newcommand{\bigO}{{\mathcal{O}}}

\newcommand{\Rnbn}{\Rbb^{n \times n}}

\newcommand{\Rmbm}{\Rbb^{m \times m}}
\newcommand{\Rmbn}{\Rbb^{m \times n}}

\newcommand{\Zn}{\Zbb^{n}}

\newcommand{\sss}{\scriptscriptstyle }

\newcommand{\A}{\boldsymbol{A}}
\newcommand{\B}{\boldsymbol{B}}

\newcommand{\D}{\boldsymbol{D}}

\newcommand{\G}{\boldsymbol{G}}
\renewcommand{\H}{\boldsymbol{H}}
\newcommand{\I}{\boldsymbol{I}}

\newcommand{\Q}{\boldsymbol{Q}}
\newcommand{\R}{\boldsymbol{R}}

\newcommand{\U}{\boldsymbol{U}}

\newcommand{\Z}{\boldsymbol{Z}}

\newcommand{\e}{\boldsymbol{e}}

\renewcommand{\v}{\boldsymbol{v}}

\newcommand{\x}{{\boldsymbol{x}}}
\newcommand{\y}{{\boldsymbol{y}}}
\newcommand{\z}{\boldsymbol{z}}
\newcommand{\0}{{\boldsymbol{0}}}

\newcommand{\br}{{\bar{r}}}

\newcommand{\bbQ}{{\bar{\Q}}}
\newcommand{\bbR}{{\bar{\R}}}

\newcommand{\hr}{{\hat{r}}}

\newcommand{\hbR}{{\hat{\R}}}

\providecommand{\round}[1]{\lfloor#1\rceil}

\usepackage{cite}
\usepackage{color}
\DeclareMathOperator*{\argmin}{arg\,min}

\begin{document}

\title{On the KZ Reduction}
\author{Jinming~Wen and Xiao-Wen~Chang
\thanks{This work was presented in part at the IEEE International Symposium on Information Theory (ISIT 2015), Hongkong.}
\thanks{J.~Wen is with  the College of Information Science and Technology and the College of Cyber Security, Jinan University, Guangzhou, 510632, China (e-mail: jinming.wen@mail.mcgill.ca).
Part of this work was done while this author was a Ph.D student
under the supervision of the second author at McGill University.}
\thanks{X.-W. Chang is with The School of Computer Science, McGill University,
Montreal, QC H3A 0E9, Canada (e-mail: chang@cs.mcgill.ca).}
\thanks{This work was partially supported by  NSERC of Canada grant 217191-17,
National Natural Science Foundation of China (No. 11871248),
``the Fundamental Research Funds for the Central Universities'' (No. 21618329)
and the postdoc research fellowship from Fonds de Recherche Nature et Technologies.}}

%


\maketitle

\begin{abstract}
The Korkine-Zolotareff (KZ) reduction  is one of the often used reduction strategies for lattice decoding.
In this paper, we first investigate  some important properties of KZ reduced matrices.
Specifically, we present a linear upper bound on the Hermit constant
which is around $\frac{7}{8}$ times of the existing sharpest linear upper bound,
and an upper bound on the KZ constant which is  {\em polynomially} smaller than the existing sharpest one.
We also propose upper bounds on the lengths of the columns of KZ reduced matrices,
and an upper bound on the orthogonality defect of KZ reduced matrices
which are even {\em polynomially and exponentially} smaller than those of boosted KZ
reduced matrices, respectively.
Then, we derive upper bounds on the magnitudes of the entries of any solution of a shortest vector problem (SVP)  when its basis matrix is LLL reduced.
These upper bounds  are useful for analyzing the complexity and understanding numerical stability
of the basis expansion  in a KZ reduction algorithm.
Finally, we propose a new KZ reduction algorithm by modifying the commonly used Schnorr-Euchner search strategy for solving SVPs and
the basis expansion  method proposed by  Zhang {\em et al.}
Simulation results show that the new KZ reduction algorithm is much faster and more numerically reliable
than the KZ reduction algorithm proposed by Zhang {\em et al.},
especially when the basis matrix is ill conditioned.
\end{abstract}



%
\IEEEpeerreviewmaketitle

\section{Introduction}
\label{s:introduction}
Given a full column rank matrix $\A\in \mathbb{R}^{m\times n}$, the lattice $\mathcal{L}(\A)$ generated by $\A$ is defined by
\beq
\label{e:latticeA}
\mathcal{L}(\A)=\{\A\x \,|\,\x \in \mathbb{Z}^n\}.
\eeq
The columns of $\A$  form a basis of $\mathcal{L}(\A)$
and $n$ is said to be the dimension of $\mathcal{L}(\A)$.
For any $n\geq2$, $\mathcal{L}(\A)$ has infinitely many  bases and any of two are connected by a unimodular matrix $\Z$
(i.e.,  $\Z \in \mathbb{Z}^{n\times n}$ satisfies $\det(\Z)=\pm1$).
More precisely, for each given lattice basis matrix
$\A\in \mathbb{R}^{m\times n}$, $\A\Z$ is also a basis matrix of $\mathcal{L}(\A)$ if and only if $\Z$ is unimodular (see, e.g., \cite{AgrEVZ02}).

The process of selecting a good basis for a given lattice, given some criterion,
is called lattice reduction.
In many applications, it is advantageous if the basis vectors are short
and close to be orthogonal \cite{AgrEVZ02}.
For more than a century, lattice reduction has been investigated by many people and several  types of reductions, such as the KZ reduction \cite{KZ73}, the Minkowski reduction \cite{Min96},
the LLL reduction \cite{LenLL82} and Seysen's reduction \cite{Sey93},
have been proposed.

Lattice reduction plays a crucial role in many  areas, such as
communications (see, e.g., \cite{Mow94} \cite{AgrEVZ02} \cite{WubSJM11}),
GPS (see, e.g., \cite{HasB98}), cryptography (see, e.g., \cite{MicR08,HanS07, HanPS11, Pei16}),
number theory (see, e.g., \cite{GolRS00} \cite{ GurSS00}), etc.
For more details, see the survey paper \cite{WubSJM11} and the references therein.
Often in these applications, a closest vector problem (CVP) (also referred to as an integer least squares problem, see, e.g., \cite{ChaWX13}) or a shortest vector problem (SVP) needs to be solved:
\beq
\label{e:CVP}
\min_{\x\in\mathbb{Z}^n}\|\y-\A\x\|_2,
\eeq
\beq
\label{e:SVP}
\min_{\x \in\mathbb{Z}^n\backslash \{\0\}} \|\A\x\|_2.
\eeq

In communications, CVP and SVP are usually solved by the sphere decoding approach.
Typically, this approach consists of
two steps. In the first step, a lattice reduction, such as the LLL reduction and KZ reduction,
is often used to preprocess the problems by reducing $\A$  or $(\A^\dag)^T$
(here $\A^\dag=(\A^T\A)^{-1}\A^T$,  the Moore-Penrose generalized inverse of $\A$, is a basis matrix
of the dual lattice of $\mathcal{L}(\A)$).
Then, in the second step, a search algorithm,
typically the Schnorr-Euchner search strategy \cite{SchE94}, which is an improvement
of the Fincke-Pohst search strategy \cite{FinP85},
is used to enumerate the integer vectors within a hyper-ellipsoid sphere (or equivalently, the lattice points within a hypersphere).
The first step, which is also called as a preprocessing step, is carried out to make  the second step faster (see, e.g., \cite{ChaWX13} \cite{AnjCK14}).

One of the most commonly used lattice reductions is the LLL reduction.
Although the worst-case complexity of the LLL reduction for reducing real lattices is not even finite
\cite{JalSM08}, the average complexity of reducing a matrix $\A$, whose entries
independent and identically follow the Gaussian distribution, is a polynomial of the rank of $\A$ ( \cite{JalSM08}, \cite{LinMH13}).
Furthermore, the LLL reduction is a polynomial time algorithm for reducing integer lattices (see \cite{LenLL82}, \cite{DauV94}).
In addition to being used as a preprocess tool in sphere decoding,
the LLL reduction is frequently used to improve the detection performance of some suboptimal detectors
 in communications \cite{ChaWX13} \cite{ WenTB16} \cite{ WenC17}.

In some communication applications, one needs to solve a sequence of CVPs, where $\y$'s are different, but $\A$'s are identical.
In this case, instead of using the LLL reduction, one usually uses the KZ reduction  to do reduction.
The reason is that although the KZ reduction is computationally more expensive than the LLL reduction,
the second step of the sphere decoding, which usually dominates the whole computational costs, becomes more efficient.
In addition to the above application, the KZ reduction has applications in solving subset sum problems
\cite{SchE94}. Moreover, it has recently been used in integer-forcing linear receiver design \cite{SakHV13} and successive integer-forcing linear receiver design \cite{OrdEN13}.

Some important properties of the KZ reduced matrices have been studied in \cite{Sch87} and \cite{LagLS90}.
For example, a quantity called the KZ constant was introduced in \cite{Sch87} to quantify the quality of
the KZ and block KZ reduced matrices, and an upper bound on the KZ constant was given in the same paper.
Upper bounds on the lengths of the columns  and on the orthogonality defect
of KZ reduced matrices were developed in \cite{LagLS90}.

There are various KZ reduction algorithms \cite{Hel85,Kan87,Sch87,AgrEVZ02,ZhaQW12}.
All of these KZ reduction algorithms involve solving SVPs and basis expansion.
Among them, the one in \cite{ZhaQW12}, which uses floating point arithmetic, is the state-of-the-art
and is more efficient than the rest.
As in \cite{AgrEVZ02},  for efficiency, the LLL reduction is employed to preprocess the SVPs
and then the Schnorr-Euchner search algorithm \cite{SchE94} is
used to solve the preprocessed SVPs in \cite{ZhaQW12}.
But instead of using Kannan's basis expansion method, which was used in \cite{Kan87} and \cite{AgrEVZ02},
it uses a new more efficient basis expansion method.
However, the algorithm has some drawbacks.
Its reduction process is slow and it is not numerically reliable
in producing a KZ reduced lattice basis, especially when the basis matrix is ill-conditioned.

In this paper, we investigate some  properties of KZ reduced matrices and propose an improved  KZ reduction algorithm to address the drawbacks of the algorithm presented in  \cite{ZhaQW12}.
The main contributions of this paper are summarized in the following:
\begin{itemize}
\item
Some important properties of  a KZ reduced matrix  are studied in this paper.
Specifically,  we first propose a linear upper bound on the Hermit constant,
which is around $\frac{7}{8}$ times of the existing sharpest one that was recently presented in  \cite[Thm.\ 3.4]{Neu17}.
Then, we develop an upper bound on the KZ constant which is polynomially smaller
than the bound given by \cite[Thm.\ 4]{HanS08}.
Furthermore, upper bounds on the lengths of the columns of
a KZ reduced triangular matrix are also presented, which are even polynomially smaller than those of
a boosted KZ reduced matrix  given in  \cite[eq.s (11-12)]{LyuL17}.
Finally, an upper bound on the orthogonality defect of a KZ reduced matrix  is provided,
which is even exponentially smaller than the one on  the orthogonality defect of  a boosted KZ reduced matrix given in \cite[eq. (13)]{LyuL17}.

\item
A simple example is given to show that the entries of a solution of a general SVP can be arbitrary large.
When the basis matrix of an SVP is LLL reduced,
an upper bound on the magnitude of each entry of a solution  of the SVP is derived.
It is  sharper than the one given in our conference paper \cite{WenC15},
which did not give a proof due to the space limitation.
The bound is not only interesting in theory, but also useful for bounding the complexity of the basis expansion,
an important step in the KZ reduction process.
Furthermore, it  provides a theoretical explanation for good numerical stability of  our modified basis expansion method (to be mentioned later).

\item
An improved Schnorr-Euchner search  algorithm for solving an SVP  is proposed.
Combining this method with our modified basis expansion method  proposed in out conference paper \cite{WenC15}
results in an improved KZ reduction algorithm.
Numerical results indicate that the new algorithm
is much more efficient and  numerically reliable than the one proposed in  \cite{ZhaQW12}.

\end{itemize}

The rest of the paper is organized as follows.
In Section \ref{s:reduction}, we introduce the LLL and KZ reductions.
In Section \ref{s:PKZ}, we investigate some vital properties of the KZ reduced matrices.
An improved KZ reduction algorithm is presented in Section \ref{s:KZ}.
Some simulation results are given in Section \ref{s:sim} to show the efficiency and numerical reliability of our new algorithm.
Finally, we summarize this paper in Section \ref{s:sum}.

{\it Notation.}
Let $\mathbb{R}^n$ and $\mathbb{Z}^n$ be the spaces of $n$-dimensional column real vectors and integer vectors, respectively.
Let $\mathbb{R}^{m\times n}$ and $\mathbb{Z}^{m\times n}$ be the spaces of $m\times n$ real matrices and integer matrices, respectively.
Boldface lowercase letters denote column vectors and boldface uppercase letters denote matrices,
e.g., $\y\in\mathbb{R}^n$ and $\A \in\mathbb{R}^{m\times n}$.
Let $\Rmbn_n$  denote the set of $m\times n$ real matrices with   rank $n$.
For $\x\in \mathbb{R}^n$, we use $\round{\x}$ to denote its nearest integer vector, i.e.,
each entry of $\x$ is rounded to its nearest integer
(if there is a tie, the one with smaller magnitude is chosen),
and use $\|\x\|_2$ to denote the 2-norm of $\x$.
For a matrix $\A$, we use $a_{ij}$ to denote its $(i,j)$ entry,
use $\A_{i:j,k}$ to denote the subvector of column $k$ with row indices from $i$ to $j$
and use $\A_{i:j,k:\ell}$ to denote the submatrix containing elements with row indices from $i$ to $j$ and column indices from $k$ to $\ell$.
Let $\e_k$   denote   the $k$-th column of an identity matrix $\I$, whose dimension depends on the context.
For $\A=(a_{ij}) \in \Rmbn$, we denote $|\A|=(|a_{ij}|)$.
For two matrices $\A,\B\in \Rmbn$, the inequality $\A \leq \B$ means $a_{ij} \leq b_{ij}$ for all $i$ and $j$.

\section{LLL and KZ reductions}\label{s:reduction}

In this section, we briefly introduce the KZ reduction.
But we first  introduce the LLL reduction, which is employed to accelerate the process of solving
SVPs, the key steps of a KZ reduction algorithm.

Let $\A$ in \eqref{e:latticeA} have the following QR factorization
(see, e.g., \cite[Chap.\ 5]{GolV13})
\beq
\label{e:QR}
\A= \Q \bmx\R \\ \0 \emx,
\eeq
where $\Q\in \Rmbm$ is orthogonal and $\R\in \Rnbn$ is nonsingular upper triangular,
and they are referred to as the Q-factor and the R-factor of $\A$, respectively.

With \eqref{e:QR}, the LLL reduction \cite{LenLL82} reduces $\R$ in \eqref{e:QR} to $\bbR$ via
\beq
\label{e:QRZ}
\bbQ^T \R \Z = \bbR,
\eeq
where $\bbQ  \in \mathbb{R}^{n\times n}$ is orthogonal,
$\Z\in   \mathbb{Z}^{n\times n}$ is  unimodular 
 and
$\bbR\in \mathbb{R}^{n\times n}$ is upper triangular and  satisfies the conditions:
for $1\leq i \leq j-1 \leq n-1$,
\begin{align}
&|\br_{ij}|\leq\frac{1}{2} |\br_{ii}|,   \label{e:criteria1} \\
&\delta\, \br_{ii}^2 \leq   \br_{ii}^2+ \br_{i+1,i+1}^2,
\label{e:criteria2}
\end{align}
where $\delta$ is a parameter satisfying $1/4 < \delta \leq 1$
The matrix $\A\Z$ is said to be  LLL reduced (or equivalently $\bbR$ is said to be LLL reduced)
and the equations \eqref{e:criteria1} and  \eqref{e:criteria2} are respectively referred to as  the size-reduced condition and the Lov\'{a}sz condition.

Similar to the LLL reduction, after the QR factorization of $\A$ (see \eqref{e:QR}),
the KZ reduction reduces $\R$ in \eqref{e:QR} to $\bbR$ through \eqref{e:QRZ},
where $\bbR$ satisfies \eqref{e:criteria1} and
\begin{align}
\label{e:criteria2KZ}
| \br_{ii}| =\min_{\x\,\in\,\mathbb{Z}^{n-i+1}\backslash \{\0\}}\|\bbR_{i:n,i:n}\x\|_2, \ \
1\leq i \leq n.
\end{align}
Then  $\A\Z$ is said to be  KZ reduced.
If $\A$'s R-factor in \eqref{e:QR} satisfies
 \eqref{e:criteria1} and \eqref{e:criteria2KZ}, i.e., they hold with $\bbR$
replaced by $\R$, then $\A$ is already KZ reduced.
Note that if a matrix is KZ reduced, it must be LLL reduced for $\delta=1$.

Combing \eqref{e:QR} with \eqref{e:QRZ}, one yields
\[
\A= \Q \bmx \bbQ & \0 \\  \0 & \I_{m-n} \emx \bmx\bbR \\ \0\emx\Z^{-1}.
\]
Since both $\Q$ and $\bmx \bbQ & \0 \\  \0 & \I_{m-n} \emx $ are orthogonal matrices, by letting $\z=\Z^{-1}\x$,
the SVP \eqref{e:SVP} can be transformed to
\beq
\label{e:SVPR}
\min_{\z\in \Zn\backslash \{\0\}}\|\bbR\z\|_2.
\eeq
Let $\z$ be a solution of the SVP \eqref{e:SVPR}, then $\Z\z$ is a solution of the SVP \eqref{e:SVP}.

\section{Some properties of the KZ reduced matrices}
\label{s:PKZ}

In this section, we investigate some properties of KZ reduced matrices.
Specifically, we present a linear upper bound on the Hermite constant,
an upper bound on the KZ constant,  upper bounds on the lengths of the columns of the KZ reduced matrices,
and an upper bound on the orthogonality defect.

\subsection{A linear upper bound on the Hermite constant}
\label{ss:HC}

Let $\lambda(\A)$ denote the length of a shortest nonzero vector in ${\cal L}(\A)$, i.e.,
$$
\lambda(\A)=\min_{\x\in \Zn\backslash \{\0\}} \|\A\x\|_2,
$$
then the Hermite constant $\gamma_n$ is defined as
\[
\gamma_n=\sup_{\A\in\mathbb{R}_n^{m\times n}}\frac{(\lambda (\A))^2}{(\det(\A^T\A))^{1/n}}.
\]

The exact values of $\gamma_n$ are only known for $n=1,\ldots,8$ \cite{Mar13}
and $n=24$ \cite{CohA04}, which are summarized in Table \ref{tb:gamma}.

\begin{table}[H]
\begin{center}
\begin{tabular}{|c|c|c|c|c|c|c|c|c|c|c|}
 \hline
$n$& $1$ & $2$ & $3$& $4$& $5$ & $6$& $7$& $8$& $24$\\ \hline
$\gamma_n$ & 1 & $\frac{2}{\sqrt{3}} $& $2^{1/3}$ & $\sqrt{2} $& $8^{1/5} $& $(\frac{64}{3})^{1/6}$
&$64^{1/7}$&2 &4\\ \hline
\end{tabular}
\end{center}
\caption{}
\vspace{-7mm}
\label{tb:gamma}
\end{table}

However, there are some upper bounds on $\gamma_n$ for general $n$.
The most well-known upper bound is probably the one obtained by Blichfeldt \cite{Bli14}:
\beq
\label{e:blichfeldt}
\gamma_n \leq \frac{2}{\pi} (\Gamma(2+n/2))^{2/n},
\eeq
where $\Gamma(\cdot)$ is a Gamma function.
For some applications,  a linear upper bound on $\gamma_n$ is useful.
For example, the inequality $\gamma_n \leq \frac{2}{3}n$ (for $n\geq 2$) \cite{LagLS90} has been used to
derive upper bounds on the lengths of the columns of the KZ reduced matrices in 
\cite[Proposition 4.2]{LagLS90}
and on the lengths of the columns and the orthogonality defect of the boosted KZ reduced matrices
in \cite[Proposition 4 and eq.\ (13)]{LyuL17}.
The inequality $\gamma_n \leq 1 + \frac{1}{4}n$ (for $n\geq 1$), which is given in \cite[p35]{NguV10} without a proof,
has been used to derive upper bounds on the proximity factors of successive interference
cancellation (SIC) decoding in \cite[eq.s (41-42)]{ZhaQW12}.

The most recent result is
\beq \label{e:Neuub}
\gamma_n \leq \frac{1}{7} n+ \frac{6}{7} \ \  \mbox{ for }   n \geq 3,
\eeq
which is presented in \cite{Neu17}.
It is stated in \cite[Thm.\ 3.4]{Neu17} that this bound can be proved
by combining \eqref{e:blichfeldt} and the fact that the inequality holds for $3\leq n\leq 36$  \cite{CohE03}.
But no detailed proof is  given there.
In the following  we give a  new linear upper bound,
which will be used to study some properties of the KZ reduction in the rest subsections.

\begin{theorem} \label{t:gamman}
For $n\geq 1$,
\beq
\label{e:gamman}
\gamma_n < \frac{1}{8}n +\frac{6}{5}.
\eeq
\end{theorem}

\begin{proof}
Since the proof is a little long, we put it in Appendix \ref{s:proofTg}.
\end{proof}

Notice that our new linear bound \eqref{e:gamman} is sharper than \eqref{e:Neuub}
when $n\geq 20$. When $n\leq 19$,  the latter  is sharper than the former,
but the difference between them is small.
By Stirling's approximation, the asymptotic value of the right-hand side of \eqref{e:blichfeldt}
is  $\frac{1}{\pi e}n\approx\frac{1}{8.54} n $. Thus, the linear bound given in \eqref{e:gamman}
is very close to it.
In fact, our linear bound \eqref{e:gamman} is very close to Blichfeldt's bound \eqref{e:blichfeldt}
not only for large $n$, but also for small $n$.
This can be clearly seen from  Figure \ref{fig:Ratio},
which displays the ratio of our new linear bound in \eqref{e:gamman} to
Blichfeldt's bound in \eqref{e:blichfeldt} for $n=2,3,\ldots, 2000$.
\begin{figure}[!htbp]
\centering
\includegraphics[width=3.2in]{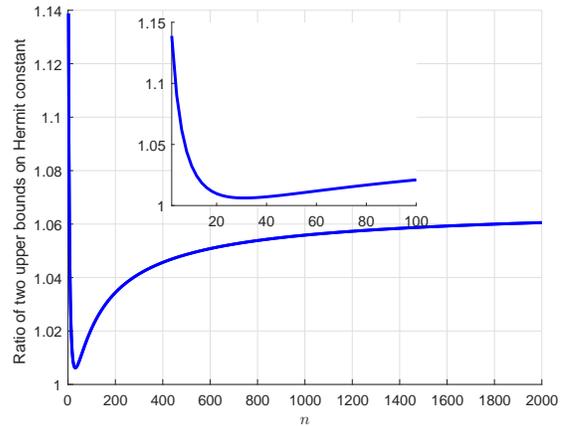}
\caption{The ratio of the bound in \eqref{e:gamman} to
Blichfeldt's bound in \eqref{e:blichfeldt}  versus $n$} \label{fig:Ratio}
\end{figure}

\begin{remark}
\label{r:LLLdr}

The following lower bound on the decoding radius of
the LLL-aided SIC decoder
is given  in  \cite[Lemma 1]{LuzSL13}:
\[
r_{\sss \mathrm{LLL-SIC}}\geq \frac{\lambda(\delta-1/4)^{(n-1)/4}}{2\sqrt{n}},
\]
where $\delta$ is the parameter of the LLL reduction (see \eqref{e:criteria2}).
By using  Table \ref{tb:gamma} and Theorem \ref{t:gamman}, we can straightforwardly
get a tighter lower bound for each $n$:
\[
 r_{\sss \mathrm{LLL-SIC}}\geq
  \begin{cases}
\frac{\lambda(\delta-1/4)^{(n-1)/4}}{2\sqrt{2}},  & 2\leq n\leq 8,\\
\frac{\lambda(\delta-1/4)^{(n-1)/4}}{2\sqrt{(5n+48)/40}}, & n>8.
  \end{cases}
\]

Note that the decoding radius $r_{\sss \mathrm{LLL-SIC}}$ of the LLL-aided SIC decoder
is the largest radius of the noise vector in the linear model  $\y=\A\x+\v$ such that
the decoder can correctly return $\x$, provided that $\|\v\|_2\leq r_{\sss \mathrm{LLL-SIC}}$.
For more details,  see \cite{LuzSL13}.
\end{remark}

\subsection{An upper bound on the KZ constant}
\label{ss:KZconstant}

In  \cite{Sch87}, the KZ constant $\alpha_n$ is defined to quantify the quality of the KZ
(and block KZ) reduced matrices.
Mathematically, $\alpha_n$ for $n$-dimensional lattices  can be expressed as
\beq
\label{e:KZconstant}
\alpha_n=\sup_{\A\in \mathcal{B}_{\sss \mathrm{KZ}}} \frac{r_{11}^2}{r_{nn}^2},
\eeq
where $\mathcal{B}_{\sss \mathrm{KZ}}$ denotes the set of all $m\times n$ KZ reduced matrices with full column rank,
and $r_{11}$ and $r_{nn}$ are the first and last diagonal entries of the R-factor of $\A$ (see \eqref{e:QR}), respectively.
Note that $|r_{11}| = \lambda(\A)$ for $\A\in \mathcal{B}_{KZ}$.

The KZ constant can be used to bound the proximity factors (see \cite[Sec.\ V-B]{Lin11})
and the lengths of the column vectors of the R-factors of KZ reduced matrices
(see \cite[Prop.\ 4.2]{LagLS90}).

Schnorr showed that $\alpha_n\leq n^{1+\ln n}$ for $n\geq 1$ \cite[Cor.\ 2.5]{Sch87} and asked whether
$\alpha_n\leq n^{O(1)}$. Ajtai gave a negative answer to this problem by showing that there is an
$\varepsilon>0$ such that $\alpha_n\geq n^{\varepsilon\ln n}$ \cite{Ajt08}.
Hanrot and Stehl{\'e} proved that \cite[Thm.\ 4]{HanS08}
\beq
\label{e:KZconstantUB1}
\alpha_n\leq n\prod_{k=2}^nk^{1/(k-1)} \leq n^{\frac{\ln n}{2}+O(1)}  \mbox{ for } n\geq 2.
\eeq

Our new upper bound on the KZ constant is stated as follows.
\begin{theorem}
\label{t:KZconstantUB}
The KZ constant $\alpha_n$ satisfies
\beq
\label{e:KZconstantUB2}
\alpha_n\leq f(n),
\eeq
where
\begin{table}[htbp!]

\begin{center}
\begin{tabular}{|c|c|c|c|c|c|c|c|c|c|c|}
 \hline
$n$& $1$ & $2$ & $3$& $4$& $5$ & $6$& $7$  & \!\! $8$ \\ \hline
\!$f(n)$\! & $1$ & $\frac{4}{3} $& $\frac{2^{\frac{3}{2}}}{ \sqrt{3}}$ & $\frac{2^{\frac{11}{6}}}{ \sqrt{3}}$&
$\frac{2^{\frac{25}{12}}}{\sqrt{3}}$& $\frac{2^{\frac{161}{60}}}{3^{\frac{7}{10}}} $ &
 $\frac{2^{\frac{161}{60}}}{3^{\frac{8}{15}}} $  \!\!\! &  $\frac{2^{\frac{1227}{420}}}{3^{\frac{8}{15}}}$  \\ \hline
\end{tabular}
\end{center}
\caption{} \label{t:KZconstant}
\end{table}
\vspace{-3mm}
\beq
\label{e:f}
f(n)=
7 \left( \frac{1}{8}n +\frac{6}{5}\right)  \left(\frac{n-1}{8}\right)^{\frac{1}{2}\ln((n-1)/8)}  \emph{ for }   n \geq 9.
\eeq
\end{theorem}

\begin{proof}
Since the proof is long, we put it in Appendix \ref{s:proofKZC}.
\end{proof}

Now we compare the first upper bound on $\alpha_n$ in \eqref{e:KZconstantUB1} with the new one 
in \eqref{e:KZconstantUB2}. When $n\geq 2$,
\begin{align*}
\prod_{k=2}^nk^{1/(k-1)}
= & \prod_{k=1}^{n-1}(k+1)^{1/k}
=   \exp\left(\sum_{k=1}^{n-1}\frac{\ln (k+1)}{k}\right) \\
 \overset{(a)}{\geq} & \exp\left(\sum_{k=1}^{n-1}\int_{k}^{k+1}\frac{\ln (t+1)}{t}dt\right)\nonumber\\
=&\exp\left(\sum_{k=1}^{n-1}\int_{k}^{k+1}\frac{\ln (1+1/t)+\ln t}{t}dt\right) \\
= & \exp\!\left(\int_{1}^{n}\frac{\ln t}{t} dt\right)\exp\!\left(\int_{1}^{n}\frac{\ln (1+1/t)}{t} dt\right)\\
\overset{(b)}{\geq}&\exp\!\left(\frac{\ln^2 n}{2}\right)
\exp\!\left(\int_{1}^{n}\frac{2}{t(2t+1)} dt\right)\\
=&\left(\frac{3n}{2n+1} \right)^2n^{\frac{1}{2}\ln n},
\end{align*}
where
(a) follows from the fact that $\frac{\ln (t+1)}{t}$ is a decreasing function of $t$ when $t\geq 1$,
and (b) is obtained by \cite[eq.\ (3)]{Top04}.
Hence, for $n\geq 9$,  the ratio of the two upper bounds on $\alpha_n$ satisfy
\begin{align}
\label{e:ratioKZ}
\frac{n\,\prod_{k=2}^nk^{1/(k-1)}}{f(n)}
\geq& \frac{n\,\left(\frac{3n}{2n+1} \right)^2n^{\frac{1}{2}\ln n}} {2n \left(\frac{n}{8}\right)^{\frac{1}{2}\ln \frac{n}{8}}} \nonumber\\
=&\frac{1}{2}\left(\frac{3n}{2n+1} \right)^2 \left(\frac{n}{2\sqrt{2}}\right)^{\ln 8}.
\end{align}
By Table \ref{t:KZconstant} and some simple calculations, one can easily check that \eqref{e:ratioKZ}
also hold for $2\leq n\leq 8$.
Thus the new bound in  \eqref{e:KZconstantUB2} is  polynomially sharper than the first upper  bound in \eqref{e:KZconstantUB1}.

In the following we make some remarks about applications of Theorem \ref{t:KZconstantUB}.

\begin{remark}
\label{r:KZpf}
By utilizing Theorem \ref{t:KZconstantUB}, we can obtain upper bounds on the
proximity factors of the KZ-aided SIC and zero forcing (ZF) decoders,
which are much sharper than the best existing
ones  given by \cite[eq.s (41) and (45)]{Lin11}.
Specifically, the inequalities
\[
\rho_{\sss \mathrm{SIC}}\leq n^{1+\ln n}, \ \ \rho_{\sss \mathrm{ZF}}\leq \left(\frac{9}{4}\right)^{n-1}n^{1+\ln n}
\]
can be respectively replaced by
\[
\rho_{\sss \mathrm{SIC}}\leq f(n) \mbox{ and }\rho_{\sss \mathrm{ZF}}\leq 1+\frac{1}{5}\left[\left(\frac{9}{4}\right)^{n-1}-1\right]f(n),
\]
where $f(n)$ is defined in Table \ref{t:KZconstant} and \eqref{e:f}.
Since the derivations are straightforward, we omit them.
\end{remark}

\begin{remark}
\label{r:KZdr}
By using Theorem \ref{t:KZconstantUB}, one can give a lower bound on  the decoding
radius of the KZ-aided SIC decoder:
\[
r_{\sss \mathrm{KZ-SIC}}\geq \frac{\lambda}{2\sqrt{f(n)}},
\]
where $f(n)$ is defined in Table \ref{t:KZconstant} and \eqref{e:f}.
The derivation is similar to that for deriving \cite[Lemma 1]{LuzSL13}, so we omit it.
\end{remark}

In addition to the applications mentioned in Remarks \ref{r:KZpf} and \ref{r:KZdr},
Theorem \ref{t:KZconstantUB} will also be used to upper bound the diagonal entries
and the lengths of the column vectors of the R-factors of KZ reduced matrices in the next subsection.

\subsection{Sharper bounds for the KZ reduced matrices}
\label{ss:KZgso}

A lattice reduction on a basis matrix is to
reduce the lengths of columns and increase the orthogonality of columns.
Thus it is interesting to obtain bounds on the lengths of the columns of the reduced basis matrix.
Results have been obtained for various reductions,
 e.g., \cite[Props.\ 1.6, 1.11, 1.12]{LenLL82} for the LLL reduction, \cite[Prop. 4.2]{LagLS90}
for the KZ reduction and \cite[Prop. 4]{LyuL17} for the boosted KZ reduction.
In this subsection, we present new bounds  for the KZ reduction, which are significantly
sharper than those in \cite[Prop. 4.2]{LagLS90}.
\begin{theorem}
\label{t:KZgso}
Suppose that $\A\in \mathbb{R}^{m\times n}_n$ is KZ reduced and its R-factor is the matrix $\R$ in \eqref{e:QR}.
Then
\beq
\label{e:gso1}
r_{11}^2\leq f(i)\,r_{ii}^2, \ \   1\leq i\leq n,
\eeq
and  more generally

\begin{align}
&  r_{ii}^2 \leq f(j-i +1)\,r_{jj}^2,\ \ 1\leq i<j\leq n, \label{e:diagUB} \\
&  \|\R_{1:i,i}\|_2^2 \leq g(i)\, r_{ii}^2,\ \  1 \leq i\leq n, \label{e:gso2}
\end{align}
where $f(i)$ is defined in Table \ref{t:KZconstant} and \eqref{e:f}, and $g(i)$ is defined as follows:
\begin{table}[htbp!]
\begin{center}
\begin{tabular}{|c|c|c|c|c|c|c|c|c|c|c|}
 \hline
$i$& $1$ & $2$ & $3$& $4$& $5$ & $6$& $7$  & \!\! $8$ \\ \hline
\!$g(i)$\! & $1$ & $1.34 $& $1.75$ & $2.27$&
$2.89$& $3.64 $ &
 $4.54$  \!\!\! &  $5.60$  \\ \hline
\end{tabular}
\end{center}
\caption{}
\label{t:gfunc}
\end{table}
\beq
\label{e:g}
g(i) = 5.6+
\frac{7(i-8)(5i+141)}{320} \left(\frac{i-1}{8}\right)^{\frac{1}{2}\ln((i-1)/8)}
\eeq
for $9 \leq i \leq n$.
\end{theorem}

\begin{proof}
By the definition of $\alpha_n$ in \eqref{e:KZconstant} and its upper bound  \eqref{e:KZconstantUB2} given in Theorem \ref{t:KZconstantUB},
we can see that \eqref{e:gso1} holds.

Since $\A$ is KZ reduced,  so are $\R_{i:j,i:j}\in \Rbb^{(j-i+1)\times (j-i+1)}$
for $1\leq i<j\leq n$. Then according to  \eqref{e:gso1}, \eqref{e:diagUB} holds.

In the following, we prove \eqref{e:gso2}.
The case $i=1$ is obvious.
We now assume $2\leq i\leq n$.
 Since $\A$ is KZ reduced, $\R$ satisfies \eqref{e:criteria1}.
Then, using \eqref{e:criteria1} and \eqref{e:diagUB}, we have
\begin{align*}
\|\R_{1:i,i}\|_2^2&=\sum_{k=1}^{i-1}r_{ki}^2+r_{ii}^2
\leq \frac{1}{4}\sum_{k=1}^{i-1}r_{kk}^2 +r_{ii}^2 \\
&\leq \big(1+\frac{1}{4}\sum_{k=1}^{i-1}f(i-k+1) \big)r_{ii}^2\\
&= \big(1+\frac{1}{4}\sum_{k=2}^{i}f(k) \big)r_{ii}^2.
\end{align*}

Set $g(i):=1+\frac{1}{4}\sum_{k=2}^{i}f(k)$ for $2 \leq i \leq 8$.
Then we use Table \ref{t:KZconstant}  to calculate $g(i)$, leading to Table \ref{t:gfunc}
(notice that each computed value has been rounded up to three decimal digits).

Now we  show that \eqref{e:gso2} holds for $ 9 \leq i \leq n$.
In fact,
\begin{align*}
 & 1+\frac{1}{4}\sum_{k=2}^{i}f(k)
=  1+\frac{1}{4}\sum_{k=2}^{8}f(k)
+\frac{1}{4}\sum_{k=9}^{i}f(k)\\
\leq \ & g(8)+\frac{1}{4} \left( \sum_{k=9}^{i}7\left(\frac{1}{8}k +\frac{6}{5}\right)\right)  \left(\frac{i-1}{8}\right)^{\frac{1}{2}\ln((i-1)/8)} \\
=\ & g(i),
\end{align*}
where in deriving  the inequality we used \eqref{e:f}.
\end{proof}

\begin{remark}
\label{r:KZconstant}
A variant of the KZ reduction called boosted KZ reduction was recently proposed in \cite{LyuL17}.
Specifically, $\A\in \mathbb{R}^{m\times n}_n$ is said to be boosted KZ reduced
if its R-factor $\R$  satisfies  \eqref{e:criteria2KZ} with $\bbR$ replaced by $\R$ and the following condition
$$
\|\R_{1:i-1,i}\|_2 \leq \|\R_{1:i-1,i}-\R_{1:i-1,1:i-1}\x\|_2, \ \forall \x\in \mathbb{Z}^{i-1}
$$
for  $2 \leq i\leq n$, i.e.,   $\|\R_{1:i-1,i}\|_2$ cannot be reduced anymore by using $\R_{1:i-1,1:i-1}$.
Suppose that $\R^{(1)}$ and $\R^{(2)}$ are respectively the KZ and boosted KZ reduced triangular matrices
reduced from the original matrix $\A$, then by the definitions of the KZ and boosted KZ
reductions, we can see that
\beq \label{e:colred}
\hspace*{-4mm}
\begin{split}
& |r^{(2)}_{ii}|=|r^{(1)}_{ii}|, \ \ 1\leq i\leq n , \\
& \|\R^{(2)}_{1:i-1,i}\|_2^2 \!\leq\! \|\R^{(1)}_{1:i-1,i}\|_2^2 \! \leq  \! \frac{1}{4}\sum_{k=1}^{i-1} (r_{kk}^{(1)})^2,\ \   2\leq i\leq n.
\end{split}
\eeq
Thus, the boosted KZ reduction is stronger than the KZ reduction in shortening the lengths
of the basis vectors.
Then it is easy to see from the definitions of boosted KZ reduction and KZ reduction
that  if $\A$ is boosted KZ reduced, \eqref{e:gso1}-\eqref{e:gso2} also hold.

For a boosted KZ reduced $\R$,  the following  bounds were presented
in \cite[eq. (11)]{LyuL17} and \cite[eq. (12)]{LyuL17}), respectively:
\begin{align}
\label{e:gso12}
r_{11}^2 & \leq \frac{8i}{9}(i-1)^{\ln(i-1)/2}\,r_{ii}^2, \\
\|\R_{1:i,i}\|_2^2 & \leq \left(1+\frac{2i}{9}(i-1)^{1+\ln(i-1)/2}\right)\, r_{ii}^2.
\label{e:gso22}
\end{align}
From the proof for the above two bounds given in \cite{LyuL17} we can see
that they also   hold when $\R$ is KZ reduced because the proof used
the inequality $ \|\R_{1:i-1,i}\|_2^2 \leq  \frac{1}{4}\sum_{k=1}^{i-1} r_{kk}^2$,
which holds for both KZ reduced $\R$ and boosted KZ reduced $\R$, see \eqref{e:colred}.
Note that \eqref{e:gso12} and \eqref{e:gso22} significantly outperform
the following upper bounds obtained in \cite[Prop. 4.2]{LagLS90} for a KZ reduced $\R$:
\begin{align*}
& r_{11}^2 \leq i^{1+\ln i}\,r_{ii}^2,\\
& \|\R_{1:n,i}\|_2^2 \leq i^{2+\ln i}\, r_{ii}^2.
\end{align*}

In the following we compare our bounds  \eqref{e:gso1}  and \eqref{e:gso2}   in Theorem \ref{t:KZgso}
with \eqref{e:gso12} and \eqref{e:gso22}, respectively.
By \eqref{e:f}, for $i\geq 9$, we have
\begin{align*}
f(i)&\leq\frac{7(5i+48)}{40}\left(\frac{i-1}{8}\right)^{\frac{1}{2}\ln(i-1)}\\
&\leq\frac{7(5+48/9)i}{40}\left(\frac{1}{8}\right)^{\frac{1}{2}\ln(i-1)}(i-1)^{\frac{1}{2}\ln(i-1)}.
\end{align*}
Thus, for $i\geq 9$, the ratio of the two upper bounds in  \eqref{e:gso12} and  \eqref{e:gso1} satisfies
\begin{align*}
\frac{\frac{8i}{9}(i-1)^{\ln(i-1)/2}}{f(i)}\geq&\frac{320  \cdot8^{\frac{1}{2}\ln(i-1)}}{63(5+48/9) }
=\frac{320}{651}(2\sqrt{2})^{\ln(i-1)}.
\end{align*}
This indicates that the upper bound  in \eqref{e:gso1} is much sharper than that in  \eqref{e:gso12}.

By \eqref{e:g}, for $i\geq 9$, we have
\begin{align*}
g(i)& = 5.6+
\frac{7(i-8)(5i+141)}{320} \left(\frac{i-1}{8}\right)^{\frac{1}{2}\ln((i-1)/8)}\nonumber\\
&\leq\frac{35i^2+707i-6104}{320}\left(\frac{i-1}{8}\right)^{\frac{1}{2}\ln (i-1)}\nonumber\\
&<\frac{35i(i-1)+(742/8)i(i-1)}{320}\left(\frac{i-1}{8}\right)^{\frac{1}{2}\ln i}\nonumber\\
&<\frac{2i(i-1)}{5}\left(\frac{i-1}{8}\right)^{\frac{1}{2}\ln i}.
\end{align*}
Thus,  for $i\geq 9$, the ratio of the two upper bounds in  \eqref{e:gso22} and  \eqref{e:gso2} satisfies
\beqnn
\frac{1+\frac{2i}{9}(i-1)^{1+\ln(i-1)/2}}{g(i)} > \frac{5}{9}(2\sqrt{2})^{\ln i}.
\eeqnn
Hence, the upper bound  in \eqref{e:gso2} is much sharper than that in  \eqref{e:gso22}.
\end{remark}

\subsection{A sharper bound on the orthogonality defect of KZ reduced matrices}
\label{ss:KZod}
One goal of performing a lattice reduction on a basis matrix is to
get a reduced basis matrix whose columns are as short as possible and as orthogonal as possible,
thus the orthogonality defect of the reduced matrices is a good measure of the quality
of the reduction.

Let $\A\in \mathbb{R}^{m\times n}_n$ be a basis matrix of a lattice.
Its orthogonality defect is defined as
\beq
\label{e:od}
\xi(\A)=\frac{\prod_{i=1}^n\|\A_{1:m,i}\|_2}{\sqrt{\det(\A^T\A)}}.
\eeq
In this following, we give an upper bound on the orthogonality defect of a KZ reduced matrix.

\begin{theorem}
\label{t:KZod}
Suppose that $\A\in \mathbb{R}^{m\times n}_n$ is KZ reduced, then
\beq
\label{e:odbd2}
\xi(\A)\leq h(n) \left(\prod_{i=1}^{n}\frac{\sqrt{i+3}}{2}\right),
\eeq
where
\begin{table}[H]
\begin{center}
\begin{tabular}{|c|c|c|c|c|c|c|c|c|c|c|}
\hline
$n$& $1$ & $2$ & $3$& $4$& $5$ & $6$& $7$& $8$  &  $n\geq 9$ \\ \hline
$h(n)$ &\!\! 1 & $\frac{2}{\sqrt{3}} $& $\sqrt{2}$ & $2$& $\sqrt{8}$& $\frac{8}{\sqrt{3}}$ &8 & 16\! \!&   $\left(\frac{1}{8}n+\frac{6}{5}\right)^{n/2}$  \\ \hline
\end{tabular}
\end{center}
\caption{}
\label{tb:h}
\end{table}
\vspace{-5mm}
\end{theorem}

\begin{proof}
By \cite[Thm. 2.3]{LagLS90}, we have
\[
\prod_{i=1}^n\|\A_{1:m,i}\|_2\leq \left(\gamma_n^{n/2}\prod_{i=1}^{n}\frac{\sqrt{i+3}}{2}\right)\sqrt{\det(\A^T\A)}.
\]
Thus
$$
\xi(\A) \leq \gamma_n^{n/2}\prod_{i=1}^{n}\frac{\sqrt{i+3}}{2}\,.
$$
Then with $\gamma_n$ given in  Table \ref{tb:gamma} for $ 1\leq n\leq 8$ and in \eqref{e:gamman} for $n \geq 9$,
we  immediately obtain \eqref{e:odbd2} with $h(n)$ given in  Table \ref{tb:h}.
\end{proof}

\begin{remark}
It was shown in \cite{LyuL17} (see  eq.(13)  there) that  for a boosted KZ reduced matrix
\beq
\label{e:odbd}
\xi(\A)\leq \frac{\sqrt{n}}{2}\left(\prod_{i=1}^{n-1}\frac{\sqrt{i+3}}{2}\right)
\left(\frac{2n}{3}\right)^{n/2},
\eeq
which is obtained based on  Minkowski's second theorem (see, e.g., \cite[VIII.2]{Cas12})
and \cite[Prop. 3]{LyuL17}.
As explained in Remark \ref{r:KZconstant},  the boosted KZ reduction is stronger than the KZ
reduction in shortening the lengths of the columns of the basis matrix.
Thus, the orthogonality defect of the matrix obtained by performing the boosted KZ reduction on
a basis matrix is not larger than that of the matrix obtained by performing the KZ reduction
on the same basis matrix.
However, from \eqref{e:odbd2}-\eqref{e:odbd} and Table \ref{tb:h},
one can see that the new upper bound on $\xi(\A)$ is about $\left(\frac{3}{16}\right)^{n/2}$ times
as small as that in \eqref{e:odbd}.
\end{remark}

\section{Upper bounds on the solution of the SVP}
\label{s:ubd}
In this section, we first give a simple example to show that some entries of the solution of
a general SVP can be arbitrarily large.
Then, we prove that when the basis matrix of an SVP is LLL reduced, all the entries of the solutions are bounded
by using a property of the LLL reduced upper triangular matrix.
The bounds are not only interesting in theory, but also useful
in analyzing the complexity of the basis expansion in the KZ reduction algorithm
(more details can be found in Sec. \ref{ss:mkz}).

The following example shows that  the entries of the solution to a general SVP can be arbitrarily large.

\begin{example}
\label{ex:unbounded}
Let $\bbR = \bmx M & M^2 & \0 \\ 0 & 1 & \0 \\   \0 & \0 & \I_{n-2}  \emx$ with $ 1< M \in \mathbb{Z}$.
Then, for any nonzero $\z \in \Zn$,
$$
\bbR\z  = [M z_1 + M^2 z_2 ,  z_2, z_3, \ldots, z_n]^T.
$$
It is easy to show that $\|\bbR\z\|_2 \ge 1$.
In fact, if $\z_{2:n} =\0$, then $z_1 \neq 0$ and  $\|\bbR\z\|_2 = |Mz_1| \geq M >1$;
otherwise, $\|\bbR\z\|_2 \geq \|\z_{2:n}\|_2 \geq1$.
Take $\z=[M, -1, 0, \ldots, 0]^T$, then $\|\bbR\z\|_2=1$.
Thus this $\z$ is a solution to the SVP  \eqref{e:SVPR}.
Since $M$ can be arbitrarily large, this $\z$ is unbounded.
\end{example}

When $\bbR$ is LLL reduced, however, we can show that all the entries of any solution
to the SVP \eqref{e:SVPR} are bounded.
Before showing that, we  present the following lemma which gives an important property of
an upper triangular matrix $\bbR$ that is size reduced.
\begin{lemma}\label{le:srub}
Let $\hat{\R} = \D^{-1} \bbR$, 
 where $\D$ is an $n\times n$ diagonal matrix with $d_{ii}=\br_{ii},1\leq i\leq n$,
and let $\U\in \Rnbn$ be an upper triangular matrix with
\beq
 \label{e:uij}
u_{ij} = \left\{ \begin{array}{ll} 1, & i=j \\ \frac{1}{2} \left(\frac{3}{2}\right)^{j-i-1},  & i<j \end{array} \right.
.
\eeq
Suppose that $\bbR$ is size reduced, i.e., \eqref{e:criteria1} holds, then
\beq
\label{e:srub}
|\hat{\R}^{-1}| \leq \U.
\eeq
\end{lemma}

This lemma is essentially the same as \cite[Lemma 2]{Lin11}  and is a special case of the
result given in the proof of \cite[Thm. 3.2]{ChaSV12},
which was easily derived by using the results given in \cite[Sec.s 8.2 and 8.3]{Hig02}.
Since its proof can be found in \cite{Lin11}, we omit its proof.

Here we make a remark.
As essentially noticed in \cite{Lin11} (see also \cite[eq.\ (8.4)]{Hig02}),
if $r_{ij} = - \frac{1}{2} r_{ii}$, then $\hr_{ij}=-\frac{1}{2}$ for $1\leq i < j \leq n$
and  the upper bound \eqref{e:srub} is attainable.

With  Lemma \ref{le:srub}, we can prove the following theorem which shows that
all the entries of any solution of an SVP, whose basis matrix is LLL reduced, are bounded.

\begin{theorem}
\label{t:zk}
Let $\z \in \mathbb{Z}^{n}$ be a solution of \eqref{e:SVPR}, where $\bbR$ is LLL reduced, then
\begin{align}
|z_i |& \leq
\sqrt{\frac{1- 2\alpha^2- \frac{1}{9} (\frac{3}{2}\alpha)^{2(n-i+1)}}{1-(\frac{3}{2}\alpha )^2}} \alpha^{i-1},\quad 1\leq i\leq n
\label{e:zkbd}
\end{align}
where
\beq
\alpha = \frac{2}{\sqrt{4\delta-1}}
\label{e:alpha}
\eeq
with $\delta$ being the parameter in the LLL reduction (see \eqref{e:criteria2}).
\end{theorem}
\begin{proof}
Since $\bar{\R}$ is LLL reduced,  by \eqref{e:criteria1} and \eqref{e:criteria2}, we have
\[
\delta\bar{r}_{ii}^2\leq \bar{r}_{i,i+1}^2+\bar{r}_{i+1,i+1}^2\leq \frac{1}{4}\bar{r}_{ii}^2+\bar{r}_{i+1,i+1}^2, \quad 1\leq i\leq n-1.
\]
Then with \eqref{e:alpha},
\[
\left | \frac{\bar{r}_{ii}}{\bar{r}_{i+1,i+1}}\right|\leq \alpha, \quad 1\leq i\leq n-1.
\]
Therefore,
\begin{align}
\label{e:ratiodiag}
\left|\frac{\bar{r}_{11}}{\bar{r}_{ii}} \right| \leq \alpha^{i-1}, \quad 1\leq i\leq n,
\end{align}
which will be used later.

Since $\z$ is a solution of \eqref{e:SVPR},
$$
\|\bbR\z\|_2 \leq \|\bbR\e_1\|_2 = |\br_{11}|.
$$
Notice that
$$
z_i = \e_i^T \z = \e_i^T \bbR^{-1} \bbR\z = \e_i^T \hat{\R}^{-1} \D^{-1} \bbR\z,
$$
where $\D$ and $\hbR$ are defined in Lemma \ref{le:srub}.
Then by the Cauchy-Schwarz inequality and Lemma \ref{le:srub}, we have
\begin{align}
|z_i|
& \leq \| \e_i^T\hat{\R}^{-1}\D^{-1} \|_2 \|\bbR\z\|_2
 \leq \| \e_i^T\hat{\R}^{-1}\D^{-1} \|_2 |\br_{11}|  \nonumber  \\
& \leq \big\|\e_i^T |\U \D^{-1} \br_{11}| \big\|_2.  \label{e:ziub}
\end{align}
Note that from \eqref{e:uij} and \eqref{e:ratiodiag},
\beq
|\U \D^{-1} \br_{11}| \leq
\bmx
1 & \frac{1}{2} \alpha & \frac{1}{2} (\frac{3}{2}) \alpha^2 & \cdot &  \frac{1}{2} (\frac{3}{2})^{n-2} \alpha^{n-1} \\
  &  \alpha & \frac{1}{2}   \alpha^2 &  &  \frac{1}{2} (\frac{3}{2})^{n-3} \alpha^{n-1} \\
  &   &  \alpha^2 & \cdot &  \frac{1}{2} (\frac{3}{2})^{n-4} \alpha^{n-1} \\
  & &   & \cdot &  \cdot  \\
   & &   &   &   \alpha^{n-1}
\emx.
\label{e:udrub}
\eeq
Then from  \eqref{e:ziub} and \eqref{e:udrub}, we obtain
\begin{align*}
|z_i|  & \leq \sqrt{ (\alpha^{i-1})^2+ \sum_{j=i+1}^{n}\left(\frac{1}{2} \Big(\frac{3}{2}\Big)^{j-i-1} \alpha^{j-1}\right)^2} \\
& =\sqrt{\frac{1- 2\alpha^2- \frac{1}{9} (\frac{3}{2}\alpha)^{2(n-i+1)}}{1-(\frac{3}{2}\alpha )^2}} \alpha^{i-1},
\end{align*}
completing the proof.
\end{proof}

The above theorem shows that when the basis matrix $\bbR$ is LLL reduced,
all the entries of any solution to \eqref{e:SVPR} are bounded and  the bounds in \eqref{e:zkbd}
depend on only the LLL reduction parameter $\delta$ and the dimension $n$.
The bounds are useful not only for analyzing the complexity of the basis expansion algorithm (see Sec.\ \ref{s:KZ})
which is a key component of the KZ reduction algorithm in \cite{ZhaQW12}, but also for understanding the advantages of our new KZ reduction algorithm
to be proposed in the next section.

Although the upper bound \eqref{e:srub} is attainable, we cannot construct  an LLL reduced upper triangular matrix $\bbR$ such that the bounds in \eqref{e:zkbd}  are reached for all $i$.
In fact, the first inequality in \eqref{e:ziub} becomes an equality
if and only  if $\bar{\R}^{-T}\e_i$ and $\bbR\z$ are linearly dependent,
which is impossible for all $i$ as $\bbR\z \neq \0$.

\section{An improved KZ reduction algorithm }\label{s:KZ}

In this section, we develop an improved KZ reduction algorithm which is much faster and more numerically
reliable than that in \cite{ZhaQW12}, especially when the basis matrix is ill conditioned.

\subsection{The KZ reduction algorithm in \cite{ZhaQW12}}

From the definition of the KZ reduction, the reduced matrix $\bbR$ satisfies both \eqref{e:criteria1} and \eqref{e:criteria2KZ}.
If $\bbR$ in  \eqref{e:QRZ} satisfies \eqref{e:criteria2KZ},
then we can easily apply size reductions to $\bbR$ such that \eqref{e:criteria1} holds.
Thus, in the following, we will only show how to obtain $\bbR$ such that \eqref{e:criteria2KZ} holds.

The algorithm needs $n-1$ steps.
Suppose that at the end of step $k-1$,  one has found an orthogonal matrix
$\Q^{(k-1)} \in \Rbb^{n\times n}$, a unimodular matrix $\Z^{(k-1)}\in \Zbb^{n\times n}$ and
an upper triangular $\R^{(k-1)}\in \Rbb^{n\times n}$ such that
\begin{align}
\label{e:recursionk-1}
(\Q^{(k-1)})^T\R\Z^{(k-1)}=\R^{(k-1)}
\end{align}
and
\beq
\label{e:diagk-1}
|r^{(k-1)}_{ii}|= \min_{\x\,\in\,\mathbb{Z}^{n-i+1}\backslash \{\0\}}\| \R^{(k-1)}_{i:n,i:n} \x\|_2,
\,\;i=1,\ldots,  k-1.
\eeq

At step $k$, as \cite{AgrEVZ02},  \cite{ZhaQW12} uses the LLL reduction aided Schnorr-Euchner
search algorithm to solve the SVP to get $\x^{(k)}$:
\begin{align}
\label{e:SVPk}
\x^{(k)}=\argmin_{\x\,\in \mathbb{Z}^{n-k+1}\setminus \{\0\}}\|\R^{(k-1)}_{k:n,k:n}\x\|_2^2.
\end{align}
Then, \cite{ZhaQW12} uses a new basis expansion algorithm to
update $\Q^{(k-1)}$ to an orthogonal $\Q^{(k)}$,
$\R^{(k-1)}$ to an upper triangular $\R^{(k)}$,
and $\Z^{(k-1)}$ to a unimodular matrix $\Z^{(k)}$ such that
\begin{align}
\label{e:recursionk}
(\Q^{(k)})^T\R\Z^{(k)}=\R^{(k)}
\end{align}
and
\beq
\label{e:diagk}
|r^{(k)}_{ii}|= \min_{\x\,\in\,\mathbb{Z}^{n-i+1}\backslash \{\0\}}\| \R^{(k)}_{i:n,i:n} \x\|_2,
 \;\,i=1,\ldots,  k.
\eeq

At the end of step $n-1$,  we get $\R^{(n-1)}$, which is just $\bbR$ in \eqref{e:QRZ} that satisfies
\eqref{e:criteria2KZ}.
Then, with  $\bbR=\R^{(n-1)}$, we can conclude that  \eqref{e:criteria2KZ} holds.

Mathematically, the basis expansion algorithm in \cite{ZhaQW12} first  constructs a  unimodular matrix
$\widetilde{\Z}^{(k)}\in \Zbb^{(n-k+1)\times (n-k+1)}$ whose first column is $\x^{(k)}$,
i.e.,
\beq
\widetilde{\Z}^{(k)} \e_1 = \x^{(k)}
\label{e:zkxk}
\eeq
and then finds an orthogonal  matrix $\widetilde{\Q}^{(k)}\in \mathbb{R}^{(n-k+1)\times (n-k+1)}$ to bring $\R^{(k-1)}_{k:n,k:n}\widetilde{\Z}^{(k)}$
back to an upper triangular matrix $\widetilde{\R}^{(k)}$, i.e., they satisfy
\begin{align*}
(\widetilde{\Q}^{(k)})^T\R^{(k-1)}_{k:n,k:n}\widetilde{\Z}^{(k)}=\widetilde{\R}^{(k)}.
\end{align*}
Let
\begin{align*}
\Q^{(k)}&=\Q^{(k-1)}\bsmx \I_{k-1}&\0\\ \0&\widetilde{\Q}^{(k)}\esmx,\\
\R^{(k)}&=\bsmx \R^{(k-1)}_{1:k-1, 1:k-1}&\R^{(k-1)}_{1:k-1, k:n}\widetilde{\Z}^{(k)}\\ \0&\widetilde{\R}^{(k)}\esmx,\\
\Z^{(k)}&=\Z^{(k-1)}\bsmx \I_{k-1}&\0\\ \0&\widetilde{\Z}^{(k)}\esmx.
\end{align*}
Then $\Q^{(k)}$ is orthogonal, $\R^{(k)}$ is upper triangular and $\Z^{(k)}$ is unimodular.
Furthermore, by \eqref{e:recursionk-1} and the above four equalities,
one can see that \eqref{e:recursionk} and \eqref{e:diagk} hold.

In the following, we introduce the process in \cite{ZhaQW12} to  obtain  $\widetilde{\Z}^{(k)}$
in \eqref{e:zkxk}. Since $\x^{(k)}$ satisfies \eqref{e:SVPk}, the greatest common divisor of all
of its entries is $1$, i.e.,
$$
\gcd(x^{(k)}_1,  x^{(k)}_2, \ldots,  x^{(k)}_{n-k+1})=1.
$$
Thus, the basis expansion algorithm in \cite{ZhaQW12} finds $\widetilde{\Z}^{(k)}$ to
transform $\x^{(k)}$ to $\e_1$ by eliminating the entries of $ \x^{(k)}$ one by one
from the last one to the second one.
Specifically, if one wants to annihilate $q$ from  $\z=[p,q]^T\in \mathbb{Z}^2$.
One can first use the extended Euclid algorithm to find two integers $a$ and $b$
such that $ap+bq=d$, where $d=\gcd(p,q)$.
Then one use $\U^{-1}$ to left multiply $\z$ to annihilate $q$ (specifically, one obtains
$\U^{-1} \z = d \,\e_1$), where the unimodular $\U$ is defined as
\beq
\label{e:U}
\U=\bmx p/d &-b\\ q/d &a  \emx.
\eeq

Based on the above explanations, the basis expansion Algorithm and the KZ reduction algorithm
in \cite{ZhaQW12} can be described in Algorithms \ref{a:expansion} and \ref{a:KZ}.

\begin{algorithm}[h!]
\caption{The Basis Expansion Algorithm in \cite{ZhaQW12}}   \label{a:expansion}
\textbf{Input:} An upper triangular $\R \in \Rnbn$, a unimodular $\Z\in \mathbb{Z}^{n\times n}$,
the index $k$ and $\x\,\in \mathbb{Z}^{n-k+1}$, a solution to the SVP\\
$\min_{\x\,\in\,\mathbb{Z}^{n-k+1}\backslash \{\0\}}\|\R_{k:n,k:n}\x\|_2.$

\textbf{Output:} The updated upper triangular $\R$ with $|r_{kk}|=\|\R_{k:n,k:n}\x\|_2$
and the updated unimodular matrix $\Z$.

\begin{algorithmic}[1]
    \FOR{$i=n-k,\dots,1$}
      \STATE find $d=\gcd(x_{i}, x_{i+1})$ and  integers $a$ and $b$ such that $ax_{i}+bx_{i+1}=d$;
      \STATE set $\U=\bmx x_{i}/d &-b\\x_{i+1}/d &a\\\emx$; \; $x_{i}=d$;
      \STATE $\Z_{1:n,k+i-1:k+i}=\Z_{1:n,k+i-1:k+i}\U$;
      \STATE $\R_{1:k+i,k+i-1:k+i}=\R_{1:k+i,k+i-1:k+i}\U$;
      \STATE find a $2\times2$ Givens rotation $\G$ such that:
      $$\G\bmx r_{k+i-1,k+i-1}\\r_{k+i,k+i-1}\\\emx=\bmx \times\\0\\\emx;$$
      \STATE $\R_{k+i-1:k+i,k+i-1:n}=\G\R_{k+i-1:i-k,k+i-1:n}$;
    \ENDFOR
\end{algorithmic}
\end{algorithm}

\begin{algorithm}[h!]
\caption{The KZ Reduction Algorithm in \cite{ZhaQW12}}   \label{a:KZ}
\textbf{Input:} A full column rank matrix $\A\in \mathbb{R}^{m\times n}$\\
\textbf{Output:} A KZ reduced upper triangular $\R \in\mathbb{R}^{n\times n}$
and the corresponding unimodular matrix $\Z\in \mathbb{Z}^{n\times n}$.

\begin{algorithmic}[1]
    \STATE compute the QR factorization of $\A$, see \eqref{e:QR};
    \STATE set $\Z=\I$;
    \FOR{$k=1$ to $n-1$}
      \STATE  solve
      $\min_{\x\,\in \mathbb{Z}^{n-k+1}\setminus \{\0\}}\|\R_{k:n,k:n}\x\|_2^2$
      by the LLL reduction-aided Schnorr-Euchner search strategy;
      \STATE apply Algorithm \ref{a:expansion} to update $\R$ and $\Z$;
    \ENDFOR
    \STATE perform size reductions on $\R$ and update $\Z$
\end{algorithmic}
\end{algorithm}

\subsection{An improved KZ reduction algorithm }
\label{ss:mkz}
In this subsection, we propose a new KZ reduction algorithm, which is much faster and
more numerically reliable than Algorithm \ref{a:KZ}, by modifying
the Schnorr-Euchner search algorithm and Algorithm \ref{a:expansion}.

First, we modify the Schnorr-Euchner search algorithm.
By \eqref{e:SVPR}, one can easily see that if $\z$ is a solution to  \eqref{e:SVPR},
then so is $-\z$.
Thus, to speed up the search, we only need to search the candidates $\z$ with $z_{n}\geq0$.
This observation was used in \cite{DinKWZ15} for integer-forcing MIMO receiver design,
which involves solving an SVP. Here we propose to extend the idea.
Note that if the solution $\z$ of  \eqref{e:SVPR} satisfies $\z_{k+1:n}=\0$ for some
$1\leq k\leq n-1$, then for efficiency, we only need to search the candidates $\z$ with $z_{k}>0$.
In this paper, we use this observation to speed up the Schnorr-Euchner algorithm.

Then, we  make a simple modification to Algorithm \ref{a:KZ}.
At step $k$, if $\x^{(k)}=\pm\,\e_1$  (see \eqref{e:SVPk}),
then obviously  Algorithm \ref{a:expansion} is not needed and we can move to step $k+1$.
Later we will come back to this observation.

In the following, we will  make some major modifications.
But before doing it,  we  introduce the following basic fact:
for any two integers $p$ and $q$, the time complexity of finding two integers $a$ and $b$ such that $ap+bq=d\equiv \gcd(p,q)$
by  the extended Euclid algorithm is bounded by $\bigO(\log_2(\min\{|p|,|q|\}))$ if  fixed precision is used \cite{Gro24}.

In Algorithm \ref{a:KZ}, after finding $\x^{(k)}$ (see \eqref{e:SVPk}), Algorithm \ref{a:expansion} is used to expand
$\R^{(k-1)}_{k:n,k:n}\x^{(k)}$ to a basis for the lattice $\{\R^{(k-1)}_{k:n,k:n}\x: \x\in \mathbb{Z}^{n-k+1}\}$.
There are some drawbacks with this approach.

\begin{itemize}
\item
Sometimes, especially when $\A$ is ill-conditioned, some of the entries of $\x^{(k)}$ may be very large
such that they are beyond the range of consecutive integers in a floating point system (i.e., integer overflow occurs),
which is very likely resulting in wrong results.
Even if integer overflow does not occur in storing $\x^{(k)}$, large $\x^{(k)}$ may cause
the problem that the computational cost of the extended Euclid algorithm is high
according to its complexity result we just mentioned before.
\item
The second problem is that  updating $\Z$ and  $\R$ in lines 4 and 5 of Algorithm \ref{a:expansion}
may cause numerical issues.
Large $x_i$ and $x_{i+1}$ are likely to produce large elements in $\U$.
As a result, integer overflow may occur in updating $\Z$,
and large rounding errors are likely to occur in updating  $\R$.
\item
Finally, $\R$ is likely to become more ill-conditioned after the updating,
making the search process for solving SVPs in later steps expensive.
\end{itemize}

In order to deal with the large $\x^{(k)}$ issue,
we look at line 4 in  Algorithm \ref{a:KZ}, which uses the LLL reduction-aided Schnorr-Euchner search algorithm to solve the SVP.
Specifically at step $k$, to solve \eqref{e:SVPk}, the LLL reduction algorithm is applied to $\R^{(k-1)}_{k:n,k:n}$:
\beq
\label{e:QRZk2}
(\widehat{\Q}^{(k)})^T\R^{(k-1)}_{k:n,k:n}\widehat{\Z}^{(k)}=\widehat{\R}^{(k-1)},
\eeq
where $\widehat{\Q}^{(k)}\in \mathbb{R}^{(n-k+1)\times(n-k+1)}$ is orthogonal, $\widehat{\Z}^{(k)}\in \mathbb{Z}^{(n-k+1)\times(n-k+1)}$
is unimodular and $\widehat{\R}^{(k-1)}$ is LLL-reduced.
Then, one solves the reduced SVP:
\begin{align}
\label{e:SVPk2}
\z^{(k)}=\argmin_{\z\,\in \mathbb{Z}^{n-k+1}\setminus \{\0\}}\|\widehat{\R}^{(k-1)}\z\|_2^2.
\end{align}
The solution of the original SVP \eqref{e:SVPk} is $\x^{(k)}=\widehat{\Z}^{(k)}\z^{(k)}$.
We will use  the improved  Schnorr-Euchner search algorithm to solve the SVPs.

Instead of expanding $\R^{(k-1)}_{k:n,k:n}\x^{(k)}$ as done in Algorithm \ref{a:KZ},
we propose to expand $\widehat{\R}^{(k-1)}\z^{(k)}$
to a basis for the lattice $\{\widehat{\R}^{(k-1)}\z: \z\in \mathbb{Z}^{n-k+1}\}$.
Unlike $\x^{(k)}$ in \eqref{e:SVPk}, which can be arbitrarily large, $\z^{(k)}$ in \eqref{e:SVPk2} is bounded (see Theorem \ref{t:zk}).

Thus, before doing the expansion, we update $\Q^{(k)}, \R^{(k)}$ and $\Z^{(k)}$ by using the LLL reduction
\eqref{e:QRZk2}:
\begin{align}
\label{e:Qk2}
\check{\Q}^{(k)}&=\Q^{(k-1)}\bsmx \I_{k-1}&\0\\ \0&\widehat{\Q}^{(k)}\esmx,\\
\label{e:Rk2}
\check{\R}^{(k)}&=\bsmx \R^{(k-1)}_{1:k-1, 1:k-1}&\R^{(k-1)}_{1:k-1, k:n}\widehat{\Z}^{(k)}\\ \0&\widehat{\R}^{(k-1)}\esmx,\\
\check{\Z}^{(k)}&=\Z^{(k-1)}\bsmx \I_{k-1}&\0\\ \0&\widehat{\Z}^{(k)}\esmx.
\label{e:Zk2}
\end{align}
Then we do basis expansion.

Now we discuss the advantages of  our modifications.
\begin{itemize}
\item
First,  the improved Schnorr-Euchner search strategy algorithm is more efficient than the original one, for more details, see the numerical simulations in Section \ref{ss:simsvp}.
\item
Second,  we expand $\widehat{\R}^{(k-1)}\z^{(k)}$ to a basis for the lattice $\{\widehat{\R}^{(k-1)}\z: \z\in \mathbb{Z}^{n-k+1}\}$,
and do not transfer $z^{(k)}$ back to  $\x^{(k)}$ as Algorithm \ref{a:KZ} does, i.e., we do not
compute $\x^{(k)}=\widehat{\Z}^{(k)}\z^{(k)}$, which can reduce some computational costs.
\item
Third, since $\widehat{\R}^{(k-1)}$ is LLL reduced, it has a very good chance,
especially when $\R$ is  well-conditioned and $n$ is small (say, smaller than 20),
that $\z^{(k)}=\pm \,\e_1$ (see \eqref{e:SVPk2}).
This was observed in our simulations.
As we stated before, the basis expansion is not needed in this case and we can move to next step
which reduces some computational costs.
\item
Finally,  the entries of $\z^{(k)}$ are bounded according to Theorem \ref{t:zk},
but the  entries of $\x^{(k)}$ may not be bounded (see Example \ref{ex:unbounded}).
Our simulations indicated that the magnitudes of the former are smaller or much smaller than
those of the latter.
Thus, the problems with using $\x^{(k)}$ for basis expansion mentioned before
can be significantly mitigated by using $\z^{(k)}$ instead.
Furthermore, by the complexity result of the extended Euclid algorithm that
we mentioned in the above, the computational costs of the basis expansion can also be reduced.
\end{itemize}

In the following, we make some further improvements.
From Algorithm \ref{a:expansion}, one can see that this basis expansion algorithm finds a sequence
of 2 by 2 unimodular matrices in the form of \eqref{e:U} to eliminate the entries of $\x$
from the last one to the second one.
Note that for any fixed $i$ (see line 1), if $x_{i+1}=0$, lines 2-7
do not need to be performed and
we only need to move to the next iteration.
In our simulations we noticed that  $\z^{(k)}$ (see \eqref{e:SVPk2}) often has a lot of zeros,
and the above modification to the basis expansion algorithm can reduce the computational cost.

Based on the above discussions, we   now present an improved  KZ reduction algorithm in Algorithm \ref{a:mKZ}.

\begin{algorithm}[h!]
\caption{An Improved KZ Reduction Algorithm}   \label{a:mKZ}
\textbf{Input:} A full column rank matrix $\A\in \mathbb{R}^{m\times n}$\\
\textbf{Output:}  A KZ reduced upper triangular $\R \in\mathbb{R}^{n\times n}$
and the corresponding unimodular matrix $\Z\in \mathbb{Z}^{n\times n}$.

\begin{algorithmic}[1]
    \STATE compute the QR factorization of $\A$, see \eqref{e:QR};
    \STATE set $\Z=\I, k=1$;
    \WHILE{$k<n$}
      \STATE compute the LLL reduction of $\R_{k:n,k:n}$ (see \eqref{e:QRZk2}) and update $\R, \Z$ (see \eqref{e:Rk2}-\eqref{e:Zk2});
      \STATE solve
      $\min_{\z\,\in \mathbb{Z}^{n-k+1}\setminus \{\0\}}\|\R_{k:n,k:n}\z\|_2^2$
      by the improved  Schnorr-Euchner search algorithm to get the solution $\z$;
      \IF{$\z=\e_1$}
      \STATE $k=k+1$;
      \ELSE
      \STATE $i=n-k$;
         \WHILE{$i\geq1$}
           \IF{$z_{i+1}\neq0$}
              \STATE  perform lines 2-7 of Algorithm \ref{a:expansion} (where $x_i$ and $x_{i+1}$
               are replaced by $z_i$ and $z_{i+1}$);
           \ENDIF
         \STATE $i=i-1$;
         \ENDWHILE
         \STATE $k=k+1$;
      \ENDIF
    \ENDWHILE
    \STATE perform size reductions on $\R$ and update $\Z$.
\end{algorithmic}
\end{algorithm}

\subsection{A concrete example}

As stated in the above subsection, Algorithm \ref{a:KZ} has numerical issues.
In this subsection, we give an example to show that Algorithm \ref{a:KZ} may not even give an LLL reduced matrix (for $\delta=0.99$),
while Algorithm \ref{a:mKZ} does.

\begin{example}
Let
$$
\A\!=\!\bmxc{rrrrr}
10.6347&-66.2715&9.3046&17.5349&24.9625\\
0&8.6759&-4.7536&-3.9379&-2.3318\\
0&0&0.3876&0.1296&-0.2879\\
0&0&0&0.0133&-0.0082\\
0&0&0&0&0.0015
\emxc.
$$
Applying Algorithm \ref{a:KZ} gives
$$
\R=\bmxc{rrrrr}
-0.2256&-0.0792&0.0125&0&0\\
0&0.2148&-0.0728&-0.0029&-0.0012\\
0&0&0.2145&0.0527&-0.0211\\
0&0&0&-0.1103&0.0306\\
0&0&0&0&0.6221
\emxc.
$$

It is easy to check that $\R$ is not LLL reduced (for $\delta=0.99$).
In fact,  $0.99\,{r}_{33}^2>{r}_{34}^2+{r}_{44}^2$.
Moreover, the matrix $\Z$ obtained by Algorithm \ref{a:KZ} is not unimodular since its determinant is $-3244032$,
which was precisely calculated by Maple.
The reason for this is that $\A$ is  ill conditioned
(its condition number in the 2-norm is about $1.0\times 10^5$) and some of the entries of $\x^{(k)}$ (see \eqref{e:SVPk}) are too large,
causing inaccuracy in updating $\R$ and integer overflow in updating $\Z$ (see lines 4-5 in Algorithm \ref{a:expansion}).

Applying Algorithm \ref{a:mKZ} to $\A$ gives
$$
{\R}=\bmxc{rrrrr}
-0.2256&0.0792&-0.0126&0.0028&-0.0621\\
0&-0.2148&0.0728&-0.0084&0.0930\\
0&0&0.2145&0.0292&-0.0029\\
0&0&0&-0.2320&0.0731\\
0&0&0&0&-0.2959
\emxc.
$$
Although we cannot verify that $\R$ is KZ reduced, we can verify that indeed it is LLL reduced.
\noindent All of the solutions of the four SVPs are $\e_1$ (note that the dimensions are different).
Thus, no basis expansion is needed.

\end{example}

\section{Numerical tests} \label{s:sim}

In this section, we do numerical tests to show the efficiencies of the improved Schnorr-Euchner
search algorithm  and the improved  KZ reduction algorithm by using the following
two classes of matrices.
\begin{itemize}
\item  Case 1. $\A$ is a $2n\times 2n$ real transformation version of the Rayleigh-fading channel matrix see, e.g., \cite{HasV05}. Specifically, let $\H=\text{randn}(n)+j\,\text{randn}(n)$,
     where $\text{randn}(n)$ is a \textsc{Matlab} built-in function,
     then
     \beq
     \label{e:A}
     \A=\bmx \Re(\H)&-\Im(\H)\\\Im(\H) &\Re(\H)\emx.
     \eeq

\item  Case 2. $\A$ is a $2n\times 2n$ real transformation version of the doubly correlated
Rayleigh-fading channel matrices, see, e.g., \cite{ChuTKV02} \cite{ShiWC08}.
Specifically, let $\H=\Psi^{1/2}(\text{randn}(n)+j\,\text{randn}(n))\Phi^{1/2}$,
where $\Psi^{1/2}$ means $\Psi^{1/2}\Psi^{1/2}=\Psi$, and both $\Psi$ and $\Phi$ are $n\times n$ matrices with $\psi_{ij}=a^{|i-j|}$ and $\phi_{ij}=b^{|i-j|}$ for $1\leq i,j\leq n$,
where $a$ and $b$ are uniformly distributed over $[0,1)$.
Then $\A$ has the form  of \eqref{e:A}.
\end{itemize}

The numerical tests were done by \textsc{Matlab} 2016b on a desktop computer with Intel(R) Core(TM)
i7-4790 CPU @ 3.60 GHz.
The \textsc{Matlab} code for Algorithm \ref{a:KZ} was provided by  Wen Zhang, one of the authors of \cite{ZhaQW12}.
The parameter $\delta$  in the LLL reduction was chosen to be 0.99.

\subsection{ Comparison of the Search Strategies}
\label{ss:simsvp}

In this subsection, we do numerical simulations to compare the efficiencies of
the original Schnorr-Euchner search algorithm developed in \cite{SchE94},
the improved one given in \cite{DinKWZ15} and our modified one proposed in Section \ref{ss:mkz}
by comparing the number of flops used by them.
These three search algorithms will be respectively denoted by ``SE-Original'',
``SE-DKWZ'' and ``SE-Improved''
in the two figures to be given in this subsection.

In the tests, for each case, for each fixed $n$, we gave 200 runs to generate 200 different $\A$'s,
resulting in 200 SVPs in the form of \eqref{e:SVP}.
Then, for each generated SVP, we use the LLL reduction to reduce
the SVP \eqref{e:SVP} to \eqref{e:SVPR} (see \eqref{e:QR} and \eqref{e:QRZ}).
Finally, we respectively solve these reduced SVPs \eqref{e:SVPR} by using the  three search algorithms.
Figures  \ref{fig:SE1} and \ref{fig:SE2} display the average number of flops taken by
the  three algorithms for solving those 200 reduced SVPs \eqref{e:SVPR} versus $n=2:2:20$
for Cases 1 and 2, respectively.

From Figures  \ref{fig:SE1} and \ref{fig:SE2}, we can see that
``SE-Improved'' is much more efficient than ``SE-DKWZ''
which is a little bit faster than ``SE-Original''  for both cases.

\begin{figure}[!htbp]
\centering
\includegraphics[width=3.2in]{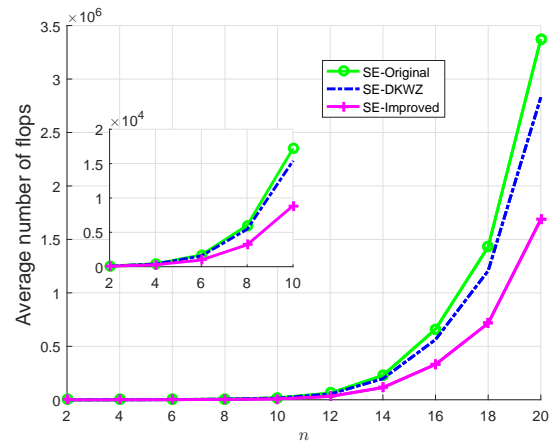}
\caption{Average number of flops taken by three Schnorr-Euchner search algorithms  versus $n$ for Case 1} \label{fig:SE1}
\end{figure}

\begin{figure}[!htbp]
\centering
\includegraphics[width=3.2in]{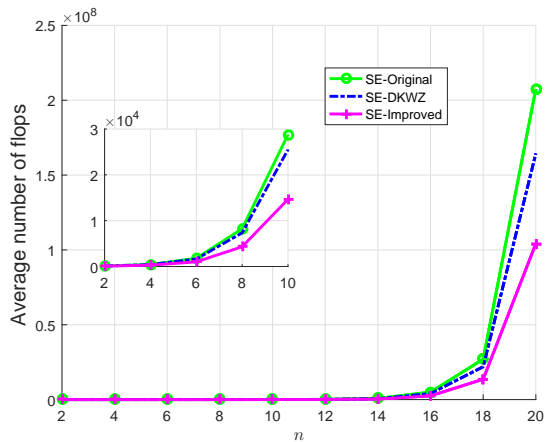}
\caption{Average number of flops taken by three Schnorr-Euchner search algorithms  versus $n$ for Case 2} \label{fig:SE2}
\end{figure}

\subsection{ Comparison of the KZ reduction algorithms} \label{ss:simkz}

In this subsection, we give numerical test results to compare the efficiencies of the
proposed KZ reduction algorithm (i.e., Algorithm \ref{a:mKZ}) and the KZ reduction
algorithm presented in \cite{ZhaQW12} (i.e., Algorithm \ref{a:KZ}).
For simplicity and clarity, the two algorithms will be referred to as ``KZ-Modified" and
``KZ-ZQW", respectively.

To see how  our new Schnorr-Euchner search algorithm and our new
basis expansion method improve  the efficiency of ``KZ-ZQW" individually,
we also compare the two KZ reduction algorithms with  the following two KZ reduction algorithms:
one is the combination of our improved Schnorr-Euchner search algorithm and the basis expansion
method proposed in \cite{ZhaQW12}, to be referred to as ``KZ-ISE'' (where ``ISE'' stand for ``improved Schnorr-Euchner'');
and the other is the combination of the Schnorr-Euchner search algorithm
proposed in \cite{SchE94} and our improved basis expansion method,
which is exactly the one proposed in our conference paper \cite{WenC15}
and will be referred to as ``KZ-WC".

In the previous subsection we compared the numbers of flops used by the three algorithms.
But here we will compare the CPU time taken by these four algorithms
because it is hard to count the flops of the extended Euclid algorithm involved in the basis expansion methods.
In our numerical tests, the \textsc{Matlab} built-in function \texttt{gcd} was used to
implement the extended Euclid algorithm.

As in the previous subsection, for each case, for each fixed $n$, we gave 200 runs to generate 200 different $\A$'s.
We then applied these four algorithms to each $\A$.
Figures  \ref{fig:CPUT1} and \ref{fig:CPUT2} display  the average CPU time of the
four algorithms over 200 runs versus $n=2:2:20$ for Cases 1 and 2, respectively.

\begin{figure}[!htbp]
\centering
\includegraphics[width=3.2in]{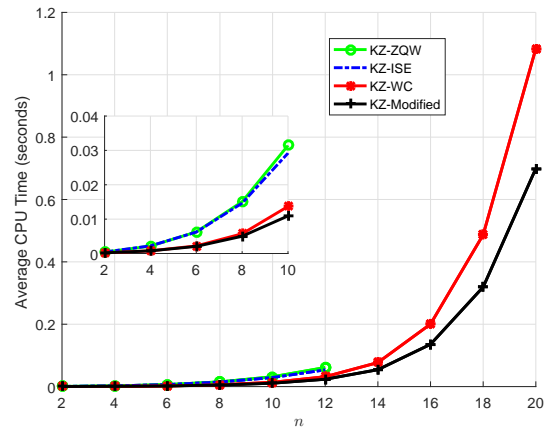}
\caption{Average CPU time of KZ reduction algorithms  versus $n$ for Case 1} \label{fig:CPUT1}
\end{figure}

\begin{figure}[!htbp]
\centering
\includegraphics[width=3.2in]{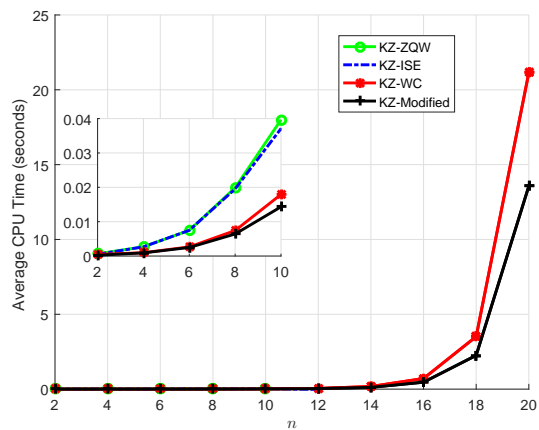}
\caption{Average CPU time of KZ reduction algorithms  versus $n$ for Case 2} \label{fig:CPUT2}
\end{figure}

Algorithms ``KZ-ZQW" and ``KZ-ISE"
often did not terminate within two hours when $n \geq 14$ for Case 1 and $n \geq 12$ for Case 2,
thus Figures \ref{fig:CPUT1} and \ref{fig:CPUT2} do  not display the corresponding results for
$n \geq 14$ and $n \geq 12$, respectively.

From Figures  \ref{fig:CPUT1} and \ref{fig:CPUT2}, we can see that for both cases ``KZ-Modified''
is faster than ``KZ-WC", which is much more efficient than two other algorithms, especially
for large $n$.
Furthermore,  in our tests we got the following warning message
from \textsc{Matlab} for ``KZ-ZQW" and ``KZ-ISE": ``Warning: Inputs contain values larger than the largest
consecutive flint. Result may be inaccurate"
for both cases and more often for Case  2 and large $n$.
This implies that the results obtained in this circumstance cannot be trusted.
But this never happened to ``KZ-WC" and ``KZ-Modified'' in our tests.
Thus the latter are more numerically reliable than the former,
as we explained in Section \ref{ss:mkz}.

Figures  \ref{fig:CPUT1} and \ref{fig:CPUT2} also indicate that the impact of our improved
basis expansion is more significant than the impact of our improved Schnorr-Euchner search algorithm
in accelerating the speed of ``KZ-ZQW".

\section{Summary} \label{s:sum}

The KZ reduction has applications in communications and cryptography.
In this paper, we have investigated some vital properties of KZ reduced matrices and
developed an improved KZ reduction algorithm.
We first developed a linear upper bound on the Hermit constant
which is around $\frac{7}{8}$ times of the upper bound given by \cite[Thm. 3.4]{Neu17},
and an upper bound on the KZ constant which is polynomially small than \cite[Thm.\ 4]{HanS08}.
We also developed  upper bounds on the columns of KZ reduced matrices,
and an upper bound on the orthogonality defect of KZ reduced matrices,
which are polynomially and exponentially smaller than  those of boosted KZ
reduced matrices given in \cite[eq.s (11-12)]{LyuL17} and \cite[eq. (13)]{LyuL17},
respectively.
Then, we presented upper bounds on the entries of any solution to an SVP
when its basis matrix is LLL reduced, while an example was given
to show that the entries can be arbitrarily large  if the basis matrix is not LLL reduced.
The bounds are useful not only for analyzing the complexity of  the extended Euclid algorithm
for the basis expansion but also for understanding the advantages of our improved KZ reduction algorithm.
Finally, we developed an improved KZ reduction algorithm by modifying the Schnorr-Euchner search strategy
and the basis expansion method.
Simulation results showed that the new KZ reduction algorithm is much more efficient and more numerically reliable than the one proposed in \cite{ZhaQW12}
especially when the bases matrices are ill conditioned.

The block KZ reduction is often used in practice as it is more efficient than
the KZ reduction and has better properties than the LLL reduction.
Some techniques have been developed to make the block KZ algorithms more efficient recently \cite{CheN11}.
We intend to apply the numerical techniques proposed  in this paper to this reduction to
improve the efficiency and numerical reliability further.
We also plan to apply the ideas developed in this paper to obtain tighter bounds for the block KZ reduction.

The Minkowski reduction, which involves solving variants of SVPs and basis expansion,
is another important reduction strategy we plan to investigate.

\section*{Acknowledgment}
We are  grateful to the editor Prof. Max Costa and the referees for their valuable and thoughtful suggestions
which significantly  improve the quality of the paper.

\appendices
\section{Proof of Theorem~\ref{t:gamman}}
\label{s:proofTg}
\begin{proof}
From Table \ref{tb:gamma}, \eqref{e:gamman} holds for $n=1$.
In the following, we assume $n\geq 2$ and prove \eqref{e:gamman}.

By \eqref{e:blichfeldt}, to show \eqref{e:gamman}, it suffices to show
\[
\left(\Gamma\left(2+\frac{n}{2}\right)\right)^{2/n} <  \frac{\pi(n+9.6)}{16}
\]
which is equivalent to
\beq
\label{e:gammansuff1}
\left(\Gamma\left(2+\frac{n}{2}\right)\right)^{2} < \left(\frac{\pi(n+9.6)}{16}\right)^{n}.
\eeq

By \cite[Thm. 1.6]{Bat08}, for $x \geq 1$
$$
\Gamma(1+x) < x^x e^{-x} \sqrt{2\pi x+ e^2 -2\pi}.
$$
Thus,
\begin{align*}
&\,\Gamma\left(2+\frac{n}{2}\right)=\left(1+\frac{n}{2}\right)\Gamma\left(1+\frac{n}{2}\right)\\
\leq& \left(1+\frac{n}{2}\right) \left(\frac{n}{2}\right)^{n/2} e^{-n/2}\left(2\pi \frac{n}{2}+ e^2 -2\pi\right)^{1/2}.
\end{align*}
Then, to show \eqref{e:gammansuff1}, we only need to show
\[
 \left(1+\frac{n}{2}\right)^2 \left(\frac{n}{2}\right)^{n} e^{-n}\left(2\pi \frac{n}{2}+ e^2 -2\pi\right)  < \left(\frac{\pi(n+9.6)}{16}\right)^{n},
\]
which is equivalent to
\beq
\label{e:phit}
\phi(t):=\frac{\left[\frac{e\pi}{8}(1+\frac{4.8}{t})\right]^{2t}}{(1+t)^2 (2\pi t+ e^2 -2\pi)} > 1
\eeq
for $t= 1, 1.5, 2, 2.5, \ldots$.

By direct calculation, one can check that
\[
\phi(t)>1, \,\;t=1, 1.5,  \ldots, 15.
\]
Thus, to show \eqref{e:phit}, we only need to show that $\phi(t)$ (or equivalently $\ln(\phi(t))$) is   increasing for $t\geq 15$.

Let $\psi(t):=\ln(\phi(t))$ and $\alpha:=e^2/(2\pi)-1$.  Then
\begin{align*}
\psi(t)=\ &  2t \ln \frac{e\pi}{8}+2t \ln \left(1+ \frac{4.8}{t}\right)  \\
 & -2 \ln (t+1) -\ln [2\pi (t+ \alpha)].
\end{align*}
The derivative of $\psi(t)$ is given by
\begin{align*}
\psi'(t)=  2\ln \frac{e\pi}{8} + 2\ln \Big(1+ \frac{4.8}{t}\Big) -  \frac{9.6}{t+4.8}
 - \frac{2}{t+1} - \frac{1}{ t+ \alpha}.
\end{align*}
Since $\ln(1+x) \geq \frac{2x}{2+x}$ for $x \geq 0$ (see, e.g., \cite[eq.\ (3)] {Top04}), for $t>0$,
\begin{align*}
\psi'(t)  \geq 2\ln \frac{e\pi}{8} \!+\! \frac{9.6}{t+2.4} \!-\! \frac{9.6}{t+4.8}  \!-\! \frac{2}{t+1} \!-\! \frac{1}{t+ \alpha}
 := \rho(t).
\end{align*}
Then
$$
\psi'(15) \geq \rho(15) = 0.0065588 \cdots >0.
$$
Thus to show $\psi(t)$ is  increasing or equivalently $\psi'(t) \geq 0$ when $t\geq 15$,
it suffices to show that $\rho(t)$ is  increasing  or equivalently $\rho'(t)>0$ when $t\geq 15$.
Note that
\begin{align*}
\rho'(t) & = - \frac{9.6}{(t+2.4)^2} +\frac{9.6}{(t+4.8)^2} + \frac{2}{(t+1)^2} + \frac{1}{(t+\alpha)^2}  \\
& > - \frac{9.6}{(t+2.4)^2} +\frac{9.6}{(t+4.8)^2} + \frac{2}{(t+2.4)^2} + \frac{1}{(t+2.4)^2}  \\
& = \frac{1}{(t+2.4)^2} \left( 9.6\frac{(t+2.4)^2}{(t+4.8)^2}-6.6\right).
\end{align*}
Here the function
\[
9.6\frac{(t+2.4)^2}{(t+4.8)^2}-6.6
\]
is  increasing with $t$ when $t\geq 0$ and its value is about $0.8138$ at $t=15$.
Thus  $\rho'(t) > 0$ when $t\geq 15$, completing the proof.
\end{proof}

\section{Proof of Theorem~\ref{t:KZconstantUB}}
\label{s:proofKZC}

To prove Theorem \ref{t:KZconstantUB}, we need  the following lemma.

\begin{lemma}\label{l:integralbd}
For $a>b>0$ and $c>0$
\beq
\label{e:integralbd}
 \int_a^b \frac{\ln( 1+c/t)}{t} d t \leq
 \frac{9}{8} \ln \frac{b(3a+2c)}{a(3b+2c)} + \frac{c(b-a)}{4ab}.
\eeq
\end{lemma}

\begin{proof}
According to \cite[eq.\ (22)]{Top04}
$$
\ln(1+x) \leq  \frac{x(6+x)}{2(3+2x)}, \ \ x \geq 0.
$$
Then, for $t> 0$, we have
$$
\frac{\ln( 1+c/t)}{t}
\leq \frac{(6 + c/t)c/t}{2(3+2c/t)t}
= \frac{3c}{(3t+2c)t} +\frac{c^2}{2(3t+2c)t^2}.
$$
Thus
\begin{align*}
 \int_a^b \frac{\ln( 1+c/t)}{t} d t
& \leq \int_a^b  \frac{3c}{(3t+2c)t} dt +  \int_a^b  \frac{c^2}{2(3t+2c)t^2} dt  \\
& = \frac{9}{8} \ln \frac{b(3a+2c)}{a(3b+2c)} + \frac{c(b-a)}{4ab}.
\end{align*}
\end{proof}

In the following, we prove Theorem \ref{t:KZconstantUB}

\begin{proof}
The case $n=1$ is trivial (note that $\alpha_1=1$).
We just assume $n\geq 2$.
By the proof of \cite[Cor. 2.5]{Sch87}, one can obtain that
\beq
\label{e:KZconstantUB3}
\alpha_n\leq \gamma_n\prod_{k=2}^n\gamma_k^{1/(k-1)}.
\eeq

For $2 \leq n \leq 8$, we   use \eqref{e:KZconstantUB3} and Table \ref{tb:gamma}
 to obtain the corresponding
upper bound   on $\alpha_n$ in Table \ref{t:KZconstant} by direct calculations.

Now we consider the case $n\geq 9$.
From  \eqref{e:KZconstantUB3}, we obtain by using \eqref{e:gamman} that
\beq
\label{e:KZconstantUB5}
\alpha_n \leq \left(\frac{1}{8} n+ \frac{6}{5}\right)\prod_{k=2}^8\gamma_k^{1/(k-1)} \prod_{k=9}^n\gamma_k^{1/(k-1)}.
\eeq
In the following we will establish bounds on the two product terms in the right hand side of \eqref{e:KZconstantUB5}.

From Table  \ref{tb:gamma}, we have
\begin{align}
\label{e:k<8}
\prod_{k=2}^8\gamma_k^{1/(k-1)}
&=  \frac{2}{\sqrt{3}} \cdot 2^{\frac{1}{6}} \cdot2^{\frac{1}{6}}
\cdot 8^{\frac{1}{20}} \cdot\left(\frac{64}{3}\right)^{\frac{1}{30}}
\cdot 64^{\frac{1}{42}} \cdot2^{\frac{1}{7}}  \nonumber\\
& = 2^{\frac{827}{420}} 3^{-\frac{8}{15}}.
\end{align}

Now we bound the second product term in the right hand side of \eqref{e:KZconstantUB5}.
Applying Theorem \ref{t:gamman}, we have
\begin{align}
& \prod_{k=9}^n\gamma_k^{1/(k-1)}   \nonumber \\
\leq &  \prod_{k=9}^n\left(\frac{1}{8}k + \frac{6}{5}\right)^{1/(k-1)}
=  \prod_{k=8}^n\left(\frac{1}{8} \left(k+\frac{53}{5}\right)\right)^{1/k}    \nonumber  \\
= & \exp\left[\sum_{k=8}^{n-1}\frac{1}{k}\ln \left( \frac{1}{8} \left(k+\frac{53}{5}\right)\right)\right]\nonumber\\
\overset{(a)}{\leq}&\exp\left(\sum_{k=8}^{n-1}\int_{k-1}^{k}\frac{1}{t} \ln  \left(\frac{1}{8}\left(t+\frac{53}{5}\right)\right)dt\right)\nonumber\\
=&\exp\!\left(\int_{7}^{n-1}\frac{1}{t} \ln\left(\frac{t+53/5}{t}\frac{t}{8}\right)dt\right)\nonumber\\
 = &\exp\!\left(\int_{7}^{n-1} \frac{1}{t} \ln \left(1 +\frac{53/5}{t}\right) dt \! \right )
 \exp\! \left(\int_{7}^{n-1}\frac{\ln(t/8)}{t}dt \! \right),   \label{e:kz2terms}
\end{align}
where (a) follows from the fact that
$\omega(t):=\frac{1}{t} \ln \left(\frac{t+53/5}{8}\right)$
is a decreasing function of $t$ when $t\geq7$,
as
\[
\omega'(t)=\frac{1}{t^2}\left(\frac{t}{t+53/5}-\ln \left(\frac{t+53/5}{8}\right) \right) <0 \ \  \mbox{ for }t\geq7.
\]

In the following we bound the two terms in the right hand side of \eqref{e:kz2terms}.
Applying Lemma \ref{l:integralbd}, we have
\begin{align}
& \exp\left( \int_{7}^{n-1}\frac{1}{t} \ln  \left(1 +\frac{53/5}{t}\right) dt \right)  \nonumber  \\
\leq \ &  \exp\left( \frac{9}{8} \ln \frac{(211/5)(n-1)}{7(3(n-1)+106/5)} +\frac{(53/5)(n-8)}{28(n-1)}  \right)  \nonumber  \\
= \ &  \exp\left(\frac{9}{8} \ln \frac{211}{105} - \frac{9}{8} \ln \left(1+ \frac{724}{105(n-1)}\right) \right. \nonumber  \\
  & + \left. \frac{53}{140}  -\frac{53}{20(n-1)} \right) \nonumber  \\
<  \ &   \left(\frac{211}{105}\right)^{9/8} \exp\left( \frac{53}{140}\right). \label{e:1stub}
\end{align}

By direct calculation, we have
\begin{align}
 & \exp\left(\int_{7}^{n-1}\frac{\ln(t/8)}{t}dt\right) \nonumber\\
=\ & \exp\left(\frac{\ln^2((n-1)/8)}{2}-\frac{\ln^2(7/8)}{2}\right)   \nonumber\\
 =\ & \left(\frac{n-1}{8}\right)^{\frac{1}{2}\ln((n-1)/8)}\left(\frac{8}{7}\right)^{\frac{1}{2}\ln(7/8)} \label{e:2ndub}
\end{align}

Then combining  \eqref{e:KZconstantUB5}-\eqref{e:2ndub} we obtain
\begin{align*}
 \alpha_n
\leq \ &  2^{\frac{827}{420}} 3^{-\frac{8}{15}}
       \left(\frac{211}{105}\right)^{9/8} \exp\left(\frac{53}{140}\right)\left(\frac{8}{7}\right)^{\frac{1}{2}\ln(7/8)} \\
&  \times  \left( \frac{1}{8}n +\frac{6}{5}\right)  \left(\frac{n-1}{8}\right)^{\frac{1}{2}\ln((n-1)/8)} \\
< \ & (6.9151 \cdots)  \left( \frac{1}{8}n +\frac{6}{5}\right)  \left(\frac{n-1}{8}\right)^{\frac{1}{2}\ln((n-1)/8)}\\
= \ & 7 \left( \frac{1}{8}n +\frac{6}{5}\right)  \left(\frac{n-1}{8}\right)^{\frac{1}{2}\ln((n-1)/8)}.
\end{align*}
\end{proof}

Here we make a remark.
In the above proof, we partitioned the indices $k$ into two parts: $2\leq k \leq 8$
and $9 \leq k \leq n$ (see \eqref{e:KZconstantUB5}).
For the first part we used the exact value of $\gamma_k$
and for the second part we used the bound \eqref{e:gamman} on $\gamma_k$.
If $n\geq 20$, we could partition the indices $k$ into three parts: $2\leq k \leq 8$, $9 \leq k \leq 19$ and $20 \leq k \leq n$,
and then for the second part we can use  \eqref{e:Neuub}, which  is sharper than \eqref{e:gamman} for this part,
as mentioned in Sec. \ref{ss:HC}.
Then a sharper bound on the KZ constant could be derived.
However, the improvement is small and the bound is complicated.
Therefore, we chose not to do it.

\bibliographystyle{IEEEtran}
\bibliography{ref}

\end{document}